%% file: main.tex
\documentclass{article}

\usepackage{amsthm}
\usepackage{amsmath}
\usepackage{amssymb}
 \usepackage{stmaryrd}
\usepackage{paralist}
 \usepackage{hyperref}

 \usepackage{times}

\theoremstyle{plain}

\newtheorem{theorem}{Theorem}[section]
\newtheorem{proposition}[theorem]{Proposition}

\newtheorem{corollary}[theorem]{Corollary}
\newtheorem{lemma}[theorem]{Lemma}

\theoremstyle{definition}

\newtheorem{definition}[theorem]{Definition}
\newtheorem{example}[theorem]{Example}
\newtheorem{remark}[theorem]{Remark}
\newtheorem{fact}[theorem]{Fact}
\newtheorem{question}[theorem]{Question}

\newcommand{\M}{\mathcal{M}}

\newcommand{\fo}{\ensuremath{\mathrm{FO}}}
\newcommand{\fotwo}{\ensuremath{\mathrm{FO^2}}}
\newcommand{\emsotwo}{\ensuremath{\mathrm{EMSO^2}}}

\newcommand{\N}{\ensuremath{\mathbb{N}}}

\newcommand{\A}{\ensuremath{\mathcal{A}}}
\newcommand{\G}{\ensuremath{\mathcal{G}}}

\newcommand{\C}{\ensuremath{\mathcal{C}}}
\renewcommand{\L}{\ensuremath{\mathcal{L}}}

\newcommand{\last}{\mathsf{last}}

\renewcommand{\S}{\ensuremath{\mathcal{S}}}
\renewcommand{\P}{\ensuremath{\mathcal{P}}}

\renewcommand{\(}{\left(}
\renewcommand{\)}{\right)}

\newcommand{\fl}{\mathrm{CL}}
\newcommand{\defeq}{\ensuremath{\stackrel{\mathrm{def}}{=}}}
\newcommand{\D}{\ensuremath{\mathcal{D}}}

\newcommand{\resp}{{\it resp.}}

\newcommand{\succclass}{{+}^c1}

\newcommand{\plusc}{\mathbin{+^c}}
\newcommand{\minusc}{\mathbin{-^c}}

\newcommand{\gsp}[1]{\mathrm{sp}\(#1\)}
\newcommand{\gdp}[1]{\mathrm{dp}\(#1\)}
\newcommand{\msp}[1]{\mathrm{msp}(#1)}

\newcommand{\sem}[1]	{[\![#1]\!]}
\newcommand{\mlabel}[1] {\label{#1}}
\newcommand{\intro}[1]{{\em #1}}

\newcommand{\DCMT}{\mathrm{CMT}}

\newcommand{\tp}[1]{\mathrm{tp}\(#1\)}

\newcommand{\firstc}{\mathsf{first}^c}
\newcommand{\lastc}{\mathsf{last}^c}
\newcommand{\firstg}{\mathsf{first}^g}
\newcommand{\lastg}{\mathsf{last}^g}

\newcommand{\true} {\mathsf{true}}

\newcommand{\nextg}{\mathtt{X}^g}
\newcommand{\nextc}{\mathtt{X}^c}
\newcommand{\prevg}{\mathtt{Y}^g}
\newcommand{\prevc}{\mathtt{Y}^c}

\newcommand{\dnextg}{\tilde{\mathtt{X}}^g}
\newcommand{\dnextc}{\tilde{\mathtt{X}}^c}
\newcommand{\dprevg}{\tilde{\mathtt{Y}}^g}
\newcommand{\dprevc}{\tilde{\mathtt{Y}}^c}

\newcommand{\comp}{\mathsf{Comp}}

\newcommand{\Fg}{\mathtt{F}^g}
\newcommand{\Fc}{\mathtt{F}^c}
\newcommand{\Gg}{\mathtt{G}^g}
\newcommand{\Gc}{\mathtt{G}^c}

\newcommand{\Pg}{\mathtt{P}^g}
\newcommand{\Pc}{\mathtt{P}^c}

\mathchardef\mhyphen="2D

\newcommand{\Hg}{\mathtt{H}^g}
\newcommand{\Hc}{\mathtt{H}^c}

\newcommand{\Ug}{\mathbin{\mathtt{U}^g}}
\newcommand{\Uc}{\mathbin{\mathtt{U}^c}}
\newcommand{\Sg}{\mathbin{\mathtt{S}^g}}
\newcommand{\Sc}{\mathbin{\mathtt{S}^c}}

\usepackage{tikz}
\usetikzlibrary{calc,matrix,chains,positioning,arrows,fit,automata,backgrounds}
\usetikzlibrary{shapes,shapes.multipart,decorations,decorations.pathmorphing,decorations.pathreplacing}
\tikzstyle{dotstyle}=[fill=black,circle,minimum size=3pt,inner sep=0pt]
\tikzstyle{automatapath}=[->,decorate,decoration={snake,segment length=5mm}]
\tikzstyle{curlybracket}=[decorate,decoration={brace,amplitude=7pt}]

\begin{document}

\title{$\mu$-calculus on data words}

\author{Thomas Colcombet and Amaldev Manuel \thanks{The research leading to these results has received funding  from the
European
Union's Seventh Framework Programme (FP7/2007-2013) under grant agreement n° 259454.}\\
{\small LIAFA, Universit\'{e} Paris-Diderot}\\
{\small \{thomas.colcombet, amal\}@liafa.univ-paris-diderot.fr}\\
}

\maketitle
  
\begin{abstract}
We study the decidability and expressiveness issues of $\mu$-calculus on data words and
data $\omega$-words. It is shown that the full logic as well as the fragment which uses only the least fixpoints are undecidable, while the fragment
containing only greatest fixpoints is decidable. Two subclasses, namely BMA and BR, obtained by limiting the compositions of formulas and their
automata characterizations are exhibited. Furthermore, Data-LTL and two-variable first-order logic are expressed as unary
alternation-free fragment of BMA. Finally basic inclusions of the fragments are discussed. 
\end{abstract}

% \setcounter{tocdepth}{4}
%\tableofcontents

%\tableofcontents

%$$\fotwo \equiv \mbox{unary-Data-LTL} \subsetneq \mbox{Data-LTL} \subsetneq \mbox{BMA} \subsetneq \mbox{BR} \subsetneq \mbox{$\nu$-fragment} \subsetneq \mbox{Data Automata}$$
\nopagebreak

\input{intro}

\input{figure}
\input{prelim.tex}
\input{mu-calculus}

\input{bounded-reversal}

\input{dltl}
\input{conclusion}

\bibliographystyle{IEEEtran}
\bibliography{mu}
\end{document}

%% file: intro.tex
% !TEX root =  main.tex

\section{Introduction}

Data words are words over the alphabet $\Sigma \times \D$ where $\Sigma$ is a finite set of {\em letters} and $\D$ is an infinite domain of {\em data values}.
Data languages are sets of such words that are invariant under permutations
of data values. This invariance reflects the fact that only properties involving
the equality of data values can be expressed in this formalism.
Typical data languages are:
\begin{itemize}
\item The first and the last data values are the same, 
\item the first data value appears a second time,
\item some data value appears twice, or its complement, all data values are different,
\item every data value at an odd position is the same as the following data value, etc\dots
\end{itemize}
This model of languages arises naturally in several contexts, such as databases or verification.

It is very desirable to extend language theory to this richer setting.
In particular, a very motivating goal is to be able to describe what should
be the natural notion of ``regular data languages''. Indeed, regular languages
of classical words form the most robust notion of language, and are basic blocks
used in the construction of many advanced results. 

However, what should be a ``regular data language''? It is not so clear since the situation
is much more complex than for word languages.
Many different formalisms
can be used for describing data languages, that can all be considered as natural extensions of regularity.
Most of them have distinct expressiveness, have different closure properties,
and different decidability status. For this reason, it is absolutely
unclear which model should be granted the name ``regular''.
Furthermore, there is no hope to find a larger class of data languages that would encompass all these
particular classes while retaining good effectiveness and decidability properties.

 Let us cite some of the most important formalisms:
\begin{asparadesc}
\item[Deterministic automata] The first and most used one is deterministic finite memory automata \cite{KaminskiFrancez94}.
	These are deterministic finite state automata that have several registers that can be used to store data values,
	and can be compared with the data value currently read. An even more ``deterministic model'' is the one of data
	monoid, which is the ``monoid variant'' of these automata \cite{Bojanczyk11data}. These models are naturally
	closed under union, intersection, and thanks to their deterministic nature, also under complement.
	Furthermore emptiness and universality are decidable properties. In exchange, these models are not very
	expressive, and deterministic finite memory automata are not closed under mirroring. Data languages recognized
	by data monoids have the same properties, and are further closed under mirroring, but these are even less expressive.
\item[Non-deterministic automata] These are the non-determini-stic counterpart of the above deterministic
model \cite{KaminskiFrancez94,KaminskiZeitlin08}. These are significantly more
	expressive, and closed under mirroring. In exchange the closure under complement and the decidability of universality are lost.
\item[Logical formalisms] The natural way to define a data language by means of a logical formula
	is to allow the use of a binary relation ``$x\sim y$''
	which signifies ``the data value at position $x$ and the data value at position $y$ are the same''. The problem is that allowing 
	this relation in first-order logic ($\fo$) immediately entails the undecidability of satisfiability. The situation is better for $\fotwo$
	(the restriction of $\fo$ to two variables, that can be reused). This class is closed under intersection, union, complement, mirroring,
	and its satisfiability is decidable \cite{BojanczykDMSS11}. The expressiveness of this model is incomparable
	to any of the above formalisms. The decidability is achieved by reduction to data automata (see below).
	By restricting the use of the new predicate ``$\sim$'' it is possible to regain decidability for logics richer than $\fotwo$. Typically
	suitable guards controlling the use of ``$\sim$'' makes monadic second-order logic equi-expressive with data-monoids	
	\cite{MFCS11:colcombet-ley-puppis}.
\item[Alternating one-way automata with one register] (of the same expressiveness as ``$\mu$-cal\-cu\-lus with free\-ze'') 
corresponds to the natural
	one-register alternating variant of the above finite memory automata \cite{DL08,JurdzinskiL11}.
	These are closed under union, intersection, complement, and emptiness and universality are decidable (but undecidable on data
$\omega$-words).
	This formalism is incomparable with all the others described in this paper.
\item[Walking models] A data word can be seen as a data structure consisting of positions, and navigational edges
	defined as follows.
	Each position is connected to its immediate successor, immediate predecessor, as well as its class successor
	and class predecessor (the class of a position is the set of positions that share the same data value; thus the class successor
	is the leftmost position to the right of the current position that carries the same data value, if it exists; the class predecessor
	is similar).
	This gives rise to models of acceptors that walk in this model, using basic commands such as ``advance to successor''
	or ``advance to the class successor''. 
	Data LTL is a member of this class \cite{KaraSZ10}. It is a variant of linear time logic (LTL) where
	operations {\em until}, {\em next}, {\em previous} and {\em since} exist in two variants, over the word and over the class. 
	An automaton mechanism, called data walking automaton (DWA), which walks on the data word is proposed in \cite{PMM}. 
	It turns out that for this model the emptiness and inclusion problems are decidable but they are strictly less expressive than data automata.
	They are not closed under projection and their closure under complementation is an open problem. The deterministic subclass, however, is closed
	under all Boolean operations.
	
\item[Data automata] Data automata were introduced  for deciding $\fotwo$ \cite{BojanczykDMSS11}.
	These are non-deterministic forms of automata, the emptiness of which is by reduction
	to reachability in petri-nets (we will encounter more precisely this model in the paper). These are closed under union and
	intersection, but not under complementation. 
\end{asparadesc}
\subsection*{Contributions}

Our contribution falls in the category of ``walking models''. In fact, we consider the most natural notion
of walking model: $\mu$-calculus. The modalities in the logic allow a formula to
refer to the predecessor, the successor, as well as the class predecessor and the class successor.
The $\mu$-calculus is well known to subsume many other formalisms, and in particular LTL. We study the properties of this
logic.

We show first that the satisfiability of the $\mu$-calculus is undecidable
(Theorem~\ref{theorem:undecidability}). For this reason, we restrict it to the $\nu$-fragment,
which is the fragment of the logic in which it is not allowed to use the least fix points.
We show that every data language definable in the $\nu$-fragment is effectively recognized by
a data automaton (Theorem~\ref{theorem:nu-to-data}). 
Furthermore, the class of languages definable in the $\nu$-fragment is naturally
closed under union, intersection, and mirroring. However it lacks closure under complement.
The previous statements carry over to the case of data $\omega$-words as well.

The second part of our analysis concerns the description of two subclasses of this logic that furthermore enjoy
the closure under complementation while retaining decidability and closure under union and intersection.
The  first such subclass is called the ``bounded reversal fragment'' (BR).
In this fragment, a fixpoint formula is allowed to switch between future modalities (``successor'' and ``class
successor'') and past modalities (``predecessor'' and ``class
predecessor'') only a bounded number of times. This class is naturally closed under complement,
and we show that it is strictly less expressive than the $\nu$-fragment (Theorem~\ref{theorem:br-to-nu}).
The decidability of BR is inherited from its inclusion in the $\nu$-fragment.
The second fragment we consider is the ``bounded mode alternation fragment'' (BMA).
In this fragment, a fixpoint formula is allowed to switch between global modalities (``successor'' and ``predecessor'')
and class modalities (``class successor'' and ``class predecessor'') only a bounded number of times. 
We show that BMA is contained in BR (Theorems~\ref{theorem:bma-to-br}).
We also show that BMA contains Data LTL, which itself contains~$\fotwo$ (Theorem~\ref{theorem:udltl-fo2}). In fact we show that
Data LTL with only unary modalities and $\fotwo$ are equivalent.
%We  show that BR and BMA can be characterized in terms of cascades of suitably defined automata
%(Proposition \ref{br-cascade} and Proposition \ref{bma-cascade}).

For the data $\omega$-word case we show that BMA is contained in data automata whereas it is not contained in the 
$\nu$-fragment. We do not treat the BR fragment for data $\omega$-words in this paper. Figures \ref{summary1}
and \ref{summary2} summarize our results. { Since all our fragments subsume $\fotwo$ their
satisfiability problems are equivalent (under elementary reductions) to reachability in vector addition systems}.

%% file: figure.tex
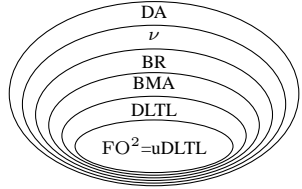
\begin{figure}
\centering\begin{tikzpicture}[scale=.7]
 \draw(0,1) ellipse (2.75cm and 1.75cm);
\draw(0,.80) ellipse (2.5cm and 1.5cm);
 \draw(0,.60) ellipse (2.25cm and 1.25cm);
\draw(0,.40) ellipse (2cm and 1cm);
 \draw(0,.20) ellipse (1.75cm and .75cm);
\draw(0,0) ellipse (1.5cm and .5cm);
  {\scriptsize
\draw (0,0) node {$\fotwo$=uDLTL};  
\draw (0,.70) node {DLTL}  ;
\draw (0,1.20) node {BMA}  ;
\draw (0,1.60) node {BR}  ;
\draw (0,2.10) node {$\nu$}  ;
\draw (0,2.55) node {DA}  ;}
\end{tikzpicture}
\caption{Decidable fragments of $\mu$-calculus on data words}
\label{summary1}
\end{figure}

\begin{figure}
\centering\begin{tikzpicture}[scale=.7]
 \draw(0,.60) ellipse (2.25cm and 1.25cm);
\draw(0,.40) ellipse (2cm and 1cm);
 \draw(0,.20) ellipse (1.75cm and .75cm);
\draw(0,0) ellipse (1.5cm and .5cm);
\draw[rotate=30](1.3,0) ellipse (.95cm and .4cm);
{\scriptsize
\draw (0,0) node {$\fotwo$=uDLTL};  
\draw (0,.70) node {DLTL}  ;
\draw (0,1.20) node {BMA}  ;
\draw (1.8,1.10) node {$\nu$}  ;
\draw (0,1.600) node {DA}  ;}
 \end{tikzpicture}
\caption{Decidable fragments of $\mu$-calculus on data $\omega$-words}
\label{summary2}
\end{figure}
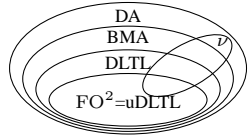

%% file: prelim.tex
\section{Preliminaries}
$\N=\{1,2,\ldots\}$ is the set of natural numbers and $+1=\{(1,2),(2,3),\ldots\}$ denotes the successor relation on $\N$. Let $\N_0 = \N \cup \{0\}$.
Denote by $[n]$ the set $\{1,\ldots,n\}$. Let $A$ be an alphabet. A word over $A$ 
is a finite sequence of letters from $A$. An $\omega$-word over $A$ is a sequence of length $\omega$ of letters from $A$.

\subsection{Data words, data $\omega$-words and data languages}
Fix a finite alphabet $\Sigma$ of \intro{letters} and an infinite set $\D$ (usually $\N$) 
of \intro{data values}.
\intro{Data words} are finite words over the alphabet $\Sigma \times \D$. 
\intro{Data $\omega$-words} are $\omega$-words over the alphabet $\Sigma \times \D$. 

Given a data word $w=(a_1, d_1)\ldots (a_n, d_n)$ (resp. data $\omega$-word $w=(a_1, d_1)(a_2,d_2)\ldots$) the
\intro{string projection} of $w$, denoted by
$\gsp{w}$, is the word $a_1\ldots a_n$ (resp. the $\omega$-word $a_1a_2\ldots$).
Similarly the
\intro{data projection} of $w$, denoted by
$\gdp{w}$, is the word $d_1\ldots d_n$ (resp. the $\omega$-word $d_1d_2\ldots$).

The data values impose a natural equivalence relation $\sim$ on the
positions of the data word (resp. data $\omega$-word), namely
$i\sim j$ if $d_i=d_j$.
For a position $i$ in $w$, the \intro{class} of $i$ is the set of all positions sharing the 
same data value as $i$. A subset $S$ of positions of $w$ is a class if it is a maximal set of 
positions sharing the same data value. Given a finite class $S=\{i_1,\ldots, i_n\}$ 
(resp. infinite class $S=\{i_1,i_2,\ldots\}$) the \intro{class projection} corresponding to 
$S$, denoted as $\gsp{w|_S}$, is the finite word $a_{i_1}a_{i_2}\ldots a_{i_n}$ (resp. the $\omega$-word 
$a_{i_1}a_{i_2}\ldots$). The class projections corresponding to each class of $w$ are collectively
called the class projections of $w$.
The set of all classes in $w$, as mentioned already, forms a partition of all the positions in the word. 
For a position $i$, the position $i+1$ is the \intro{successor} of $i$ and the position $i-1$ is the \intro{predecessor} of $i$.
We say the position $j$ is the \intro{class successor} of $i$ or $i$ is the \intro{class predecessor} of $j$, denoted as $i \plusc1= j$ or
$j\minusc1=i$, if $j$ is
the least position after position $i$ having 
the same data value.  

We denote by $\M$ the finite alphabet $\{\P,\neg \P\} \times \{\S,\neg \S\}$ called the \intro{marking alphabet}.
Given a position $i$ the \intro{1-type} (or simply \intro{type}) $\tp{i}\in \M$ of $i$ is defined as follows; 
$\tp{i}=(p,s)$ where $s = \S$ if $i$ is not the last position (if it exists) and $i+1 = i\plusc1$, and 
$\neg \S$ otherwise. Similarly $p = \P$ if $i$ is not the first position and 
$i-1 = i\minusc 1$, and $\neg \P$ otherwise. The \intro{marked string projection} of $w$,
denoted as $\msp{w}$, is the word $(a_1, \tp{1}) \ldots (a_n, \tp{n})$ (resp. the $\omega$-word
$(a_1, \tp{1}) (a_2, \tp{2})\ldots$) over the alphabet $\Sigma \times \M$. 

Given a finite class $S=\{i_1,\ldots, i_n\}$ 
(resp. infinite class $S=\{i_1,i_2,\ldots\}$) the \intro{marked class projection} corresponding to 
$S$, denoted as $\msp{w|_S}$, is the finite word $(a_{i_1}, \tp{i_1}) (a_{i_2},\tp{i_2}), \ldots (a_{i_n},\tp{i_n})$ (resp. the $\omega$-word 
$(a_{i_1},\tp{i_1})(a_{i_2},\tp{i_2})\ldots$). The marked class projections corresponding to each class of $w$ are collectively
called the marked class projections of $w$.

Let $\pi:\D \rightarrow \D$ be a permutation of $\D$. 
The permutation of $w$ under $\pi$ is 
defined to be the data word $(a_1,\pi(d_1))\ldots (a_n,\pi(d_n))$ (resp. the data 
$\omega$-word $(a_1,\pi(d_1))(a_2,\pi(d_2))\ldots$). A language of
data words $L \subseteq \(\Sigma\times \D\)^*$ is a set of data words such that
for every data word $w$ and every permutation $\pi$ of $\D$, $w\in L$ if and only 
if $\pi(w) \in L$. Similarly a language of
data $\omega$-words $L \subseteq \(\Sigma\times \D\)^\omega$ is a set of data $\omega$-words 
such that for every data $\omega$-word $w$ and every permutation $\pi$ of $\D$, $w\in L$ 
if and only if $\pi(w) \in L$. A consequence of such an invariance is that 
as far as a model of computation on data words which defines a data language is concerned 
individual data values are not important but only the relationship they induce on the positions
(namely the class relations). This is formalized as follows.
To each $w$ we associate the graph $G_w=\(D, \ell, +1, \plusc1\)$ where $D$ is 
the set of all positions in $w$ (i.e. $[n]$ if $w$ is finite and $\omega$ otherwise), $\ell:\Sigma \rightarrow 2^D$ is the labelling function
defined as $\ell(a) = \{i \mid a_i=a\}$, $+1$ is the successor relation on $\N$
restricted to $D$, and $\plusc1$ is the class successor relation of $w$. Henceforth
we will identify a data word with its graph.

Given a subset $S$ of $D$ we define
\begin{align*}
S-1 &= \{i-1 \in D \mid i \in S\} &S\minusc 1 &=\{i\minusc 1 \in D \mid i \in S\}\\
S+1 &=\{i+1 \in D \mid i \in S\}  &S\plusc 1 &=\{i\plusc 1 \in D \mid i \in S\}\\ 
\end{align*}

\begin{example}
The example shows a finite data word and its corresponding graph. Dotted and thick arrows denote the successor and class successor 
functions respectively.

\begin{center}	
	\begin{tikzpicture}[scale=.8]
% 		\node (EES-desc) at (8,5) {$\EES := (\{0,1\}^*, S_0, S_1, \preceq, \el)$};
% 		\node (CEES-desc) at (8,4.5) {$\CEES := (\{0,1\}^*, S_0, S_1, \preceq, \el, D_0, D_1)$};
% 		
		\node (a1) at (1,0) {$\begin{array}{l}a\\ 1 \end{array}$};
		\node (a2) at (2.5,0) {$\begin{array}{l}b\\ 2 \end{array}$};
		\node (a3) at (4,0) {$\begin{array}{l}a\\ 2 \end{array}$};
		\node (a4) at (5.5,0) {$\begin{array}{l}a\\ 1 \end{array}$};
		\node (a5) at (7,0) {$\begin{array}{l}b\\ 3 \end{array}$};
		\node (a6) at (8.5,0) {$\begin{array}{l}a\\ 1 \end{array}$};
		\node (a7) at (10,0) {$\begin{array}{l}b\\ 2 \end{array}$};
		
		\draw[->,dotted, thick] (a1) to (a2);
		\draw[->,dotted,thick] (a2) to (a3);
		\draw[->,dotted,thick] (a3) to (a4);
		\draw[->,dotted,thick] (a4) to (a5);
		\draw[->,dotted,thick] (a5) to (a6);
		\draw[->,dotted,thick] (a6) to (a7);
		
		\draw[->,thick] (a1) to[out=25,in=155] (a4);
		\draw[->,thick] (a4) to[out=25,in=155] (a6);
	
		\draw[->,thick] (a2) to[out=25,in=155] (a3);
		\draw[->,thick] (a3) to[out=25,in=155] (a7);

	\end{tikzpicture}
\end{center}     
The first position has type $(\neg \P, \neg \S)$, while the second position has type $(\neg \P, \S)$.
\end{example}

\intro{Two-variable first order logic} 
(in short $\fotwo$) over data words (resp. data $\omega$-words) 
is the first order logic with two variables $x$ and $y$ with predicates $a(x)$ (\textit{the position is labelled by $a$}), $x=y$, $x<y$, $x+1=y$,
$x\plusc1=y$, and $x<^c y$ (where $<^c$ is the transitive closure of $\plusc1$). Note that
$x\sim y$ is definable in $\fotwo$ in terms of $x<^c y$.
\intro{Existential MSO with two-variable kernel} (in short $\emsotwo$) is the set of all formulas of the form 
$\exists X_1\ldots \exists X_k~ \varphi$ where $\varphi$ is a $\fotwo$ formula over data words.

\subsection{Data automata and Data $\omega$-automata}

A \intro{data automaton} $\A=(B, \Sigma', C)$ is a composite automaton consisting of a non-deterministic letter-to-letter finite
state transducer $B$ with input alphabet $\Sigma \times \M$ and output alphabet $\Sigma'$, and a finite state
automaton $C$ with input alphabet $\Sigma'$. 
On a data word $w$ the automaton $\A$ work as follows. 
The transducer $B$ runs over the word $\msp{w}$ and outputs a string $v'\in \Sigma'^*$ if 
the run succeeds. Let $w'$ be the unique data word such that $\gsp{w'}=v'$ and $\gdp{w'}=\gdp{w}$. 
(Note that the fact that the transducer is length preserving is crucial here).
For each class $S$ in $w'$, the automaton $C$ runs over 
the word $\gsp{w'|_S}$. The automaton $\A$ accepts $w$ if all the runs are successful. 

A \intro{data $\omega$-automaton} (abbreviated as DA) $\A=(B, \Sigma', C, C_\omega)$ is a composite 
automaton consisting of a non-deterministic letter-to-letter finite
state \intro{B\"{u}chi} transducer $B$ with input alphabet $\Sigma \times \M$ and output alphabet $\Sigma'$, a
finite state automaton $C$ with input alphabet $\Sigma'$ and a finite state B\"{u}chi automaton
$C_\omega$ over the alphabet $\Sigma'$. On a data $\omega$-word $w$ the automaton $\A$ work as follows. 
The transducer $B$ runs over the $\omega$-word $\msp{w}$ and outputs a string $v'\in \Sigma'^\omega$ if 
the run succeeds. Let $w'$ be the unique data $\omega$-word such that $\gsp{w'}=v'$ and $\gdp{w'}=\gdp{w}$. 
For each finite class $S$ in $w'$, the automaton $C$ runs over 
the word $\gsp{w'|_S}$ and for each infinite class $S$ in $w'$, the automaton $C_\omega$ runs over 
the $\omega$-word $\gsp{w'|_S}$. The automaton $\A$ accepts $w$ if all the runs are successful. 

The most remarkable thing about data automata is that,
\begin{theorem}[\cite{BojanczykDMSS11}]
 Emptiness problem for data automata and data $\omega$-automata is elementarily equivalent to
 the reachability problem for vector addition systems and hence is decidable. 
\end{theorem}

It is a consequence of Hanf's theorem (for two-variable logic) that data automata and data $\omega$-automata are equivalent
to $\emsotwo$ with predicates $a(x)$, $x=y$, $x+1=y$, and $x\plusc 1=y$. However with a more
intricate analysis it can be shown that,

\begin{theorem}[\cite{BojanczykDMSS11}]
\label{theorem:dataemso2}
Data automata and data $\omega$-automata are equivalent to $\emsotwo$ over data words. 
\end{theorem}

%% file: mu-calculus.tex
% !TEX root =  main.tex

\section{\mbox{$\mu$-Calculus on Data Words}}
\mlabel{section:mu-calculus}
In this section, we introduce $\mu$-calculus over data words and data $\omega$-words and establish the basic decidability results.

Let $\mathit{Prop}=\{p,q, \ldots\}$ and $\mathit{Var}=\{x, y, \ldots\}$ be countable sets of propositional variables and
fixpoint variables respectively. The $\mu$-calculus on data words is the set of all formulas $\varphi$ given by the following syntax.
\begin{align*}
 \mathtt{M} &:= \nextg \mid \nextc \mid  \prevg \mid \prevc\\
A &:=  p \in \mathit{Prop} \mid \S \mid \P \mid \firstc \mid \firstg \mid \lastc \mid \lastg\\
\varphi &:= x \mid A \mid \neg A \mid \mathtt{M}\,\varphi \mid \varphi \vee \varphi \mid \varphi \wedge  
\varphi \mid \mu x. \varphi \mid \nu x. \varphi
\end{align*}

Next we disclose the semantics; as usual, on a given structure each formula denotes the set of positions where it is true. The modality $\S$ is true
at a position $i$ if the successor and class successor of $i$ coincide. Similarly
$\P$ is true at $i$ if the predecessor and class predecessor of $i$ coincide. 
The modalities $\nextg \varphi$, $\nextc\varphi$, $\prevg\varphi$, $\prevc\varphi$ 
hold if $\varphi$ holds on the successor, class successor, predecessor and class predecessor positions respectively.

\begin{figure*}[ht]
\begin{align*}
\sem{x}_w &=  \ell(x) &  & \\
\sem{\firstg}_w &=  \{1\} &\sem{\lastg}_w &=  \{ n \in D \mid  \forall i \in D~ n \geq i \}  \\
\sem{\firstc}_w &=  \{i \mid \nexists j = i \minusc 1\} &\sem{\lastc}_w &= \{i \mid \nexists j = i \plusc 1\}  \\
\sem{p}_w &=  \ell(p) &\sem{\neg p}_w &=  D\setminus \ell(p)  \\
\sem{\nextg \varphi }_w	&= \sem{ \varphi }_w -1 & \sem{\nextc \varphi }_w &=  \sem{ \varphi }_w \minusc 1\\
\sem{\varphi_1 \wedge \varphi_2 }_w &=  \sem{ \varphi_1 }_w \cap \sem{ \varphi_2 }_w &
\sem{\varphi_1 \vee \varphi_2 }_w &=  \sem{ \varphi_1 }_w \cup \sem{ \varphi_2 }_w \\
\sem{\prevg \varphi }_w	&=  \sem{ \varphi }_w +1 &
	\sem{\prevc \varphi }_w	 &= \sem{ \varphi }_w \plusc 1\\
\sem{\P}_w &=  \{i ~\mid~ i-1 = i \minusc 1 \} &
\sem{\S}_w  &=  \{i ~\mid~ i+1 = i\plusc 1\}\\
\sem{\mu x. \varphi}_w 	&=  \bigcap \Big\{S \subseteq D~\mid~ \sem{ \varphi }_{w[\ell(x):=S]} \subseteq
S\Big\}&
\sem{\nu x. \varphi}_w	&=  \bigcup \Big\{S \subseteq D ~\mid~  S \subseteq \sem{\varphi}_{w[\ell(x):=S]}\Big\} 
\end{align*}
\caption{Semantics of $\mu$-calculus on a data word ($\omega$-word) $w=\(D, +1, \succclass, \ell\)$.\label{figure:semantics}}
\end{figure*}

Note that we allow negation only on atomic propositions. However it is possible to negate a formula in our logic.
For this we define the dual modalities $\dnextg$, $\dprevg$, $\dnextc$, $\dprevc$ of $\nextg, \prevg, \nextc, \prevc$ respectively  
and the following relationship holds between them. Take special note that below $\neg$ means set complement.
$$\begin{array}{llllll}
\nextg \varphi &\equiv& \neg \dnextg \neg \varphi & \nextc \varphi &\equiv& \neg \dnextc \neg \varphi\\
\prevg \varphi &\equiv& \neg \dprevg \neg \varphi & \prevc \varphi &\equiv& \neg \dprevc \neg \varphi\\
\end{array}$$

Since the class successor relation is functional (a relation $R$ is functional if for every $x$ in the domain of $R$ there is at most one $y$ 
such that $xRy$), on all positions $i$ with a class successor, the formula $\dnextc \varphi$ is true if and only if 
$\nextc \varphi$ is true. On the other hand on all positions $i$ which do not have a class successor, $\dnextc \varphi$ is true
while $\nextc \varphi$ is false. Hence in $\dnextc$ is equivalent to 
$$\dnextc \varphi \equiv \lastc \vee \nextc \varphi$$
Since all relations in our graph are functional, similar relationship holds between 
all our modalities and their duals i.e.,
\begin{align*}
\dnextg \varphi \equiv \lastg \vee \nextg \varphi\ ,\  \dprevg \varphi \equiv \firstg \vee \prevg \varphi\ ,\  \dprevc \varphi \equiv \firstc \vee
\prevc \varphi
\end{align*}
Coming to the fixpoint formulas, each formula
$\varphi(x)$ defines a function from sets of positions
to sets of positions which is furthermore monotone (since we do not allow negation of variables). Hence by Knaster-Tarski theorem (which says that
{\em fixpoints of a monotone function on a complete lattice form a complete lattice}) it has 
fixpoints. In particular it has a least fixpoint which is intersection of all pre-fixpoints (a set of positions $S$ such that $\varphi(S) \subseteq
S$)
and a greatest fixpoint which is the union of all post-fixpoints (a set of positions $S$ such that $S \subseteq \varphi(S)$). We define 
the denotation of $\mu x. \varphi(x)$ and $\nu x. \varphi(x)$ to be the least and greatest fixpoints of $\varphi(x)$. Finally we note that the
following holds:
$$\mu x. \varphi(x) \equiv \neg \nu x. \neg \varphi(\neg x)\ .$$
The formal semantics $\sem\varphi_w$ of a formula $\varphi$ over a data word~$w$  is described in Figure~\ref{figure:semantics}.

To negate a formula $\varphi$ we take the dual of $\varphi$;
this means exchanging in the formula $\wedge$ and $\vee$, $\mu$ and $\nu$,
$p$ and $\neg p$, and all the modalities with their dual. This allows us to talk of 
$\neg \varphi$ even when $\varphi$ is not atomic, so far as the particular fragment $\varphi$ is in has 
all the necessary fixpoint operators and modalities to take the dual.

Next we lay out some terminology and abbreviations which we will use in the subsequent sections.
Let $\lambda$ denote either $\mu$ or $\nu$. Every occurrence of a fixpoint variable $x$ in a subformula $\lambda x. \psi$ of a formula is called
\intro{bound}. 
All other occurrences of $x$ are called \intro{free}.
A formula is called a \intro{sentence} if all the fixpoint variables in $\varphi$ are bound.
If $\varphi(x_1, \ldots, x_n)$ is a formula with free variables $x_1,\ldots, x_n$, then by $\varphi(\psi_1, \ldots, \psi_n)$ we mean the formula
obtained by substituting $\psi_i$ for each $x_i$ in $\varphi$.  
As usual the bound variables of $\varphi(x_1,\ldots, x_n)$ may require a renaming to avoid the capture of the free variables of $\psi_i$'s.
For a sentence 
$\varphi$ and a position $i$ in the word $w$, we denote by $w,i\models \varphi$ if $i\in\sem{\varphi}_w$.
The notation $w\models\varphi$ abbreviates the case when $i=1$.
The \intro{data language} of a sentence $\varphi$ is the set of data words $w$ such that $w\models\varphi$,
while the \intro{data $\omega$-language} of a sentence $\varphi$ is the set of data $\omega$-words $w$ such that $w\models\varphi$,

By $\mu$-fragment we mean the subset of $\mu$-calculus
which uses only $\mu$ fixpoints. Similarly $\nu$-fragment stands for the subset which uses only $\nu$-fixpoints. 

%\todo{Requires dual of modalities}
\begin{example}[Temporal modalities]
\label{modalityexample}
 An example of a formula would be $\varphi\Ug\psi$
which holds if $\psi$ holds in the future, and $\varphi$ holds in between.
This can be implemented as $\mu x.\psi\vee(\varphi\wedge\nextg x)$
The formula  $\varphi\Uc\psi=\mu x.\psi\vee(\varphi\wedge\nextc x)$
is similar, but for the fact that it refers only to the class of the current position.
The formula $\Fg\varphi$ abbreviates $\top\Ug\varphi$, and its dual
is $\Gg\varphi=\neg\Fg\neg\varphi$.
The constructs $\Sg$, $\Sc$, $\Pg$, $\Pc$, $\Hg$ and $\Hc$,
are defined analogously, using past modalities, and correspond
respectively to $\Ug$, $\Uc$, $\Fg$, $\Fc$, $\Gg$ and $\Gc$.
For instance, $\Fc\Pc\varphi$ expresses that there is a position
in the class that satisfies $\varphi$
and $\Fc\Pc(\varphi\wedge\dnextc\Gc\neg\varphi\wedge\dprevc\Hc\neg\varphi)$
expresses that there exists exactly one position which satisfies $\varphi$ in the class.
\end{example}

\begin{example} 
 The formula $\Gg\Fg (\firstc \wedge \nu x. \nextc x)$ is satisfied by all data $\omega$-words that have
 infinitely many infinite classes. Its negation $\Fg\Gg (\firstc \rightarrow \mu x. \dnextc x)$ says that 
  eventually all classes are of finite length (but still there could be infinite classes in the word). The formula $\Gg\Fg (\firstc \wedge
\mu x. \dnextc x)$ says that there exist
 infinitely many finite classes.
\end{example}

We say a variable $x$ in 
$\lambda x. \varphi(x)$ is {\em guarded} if each occurrence of $x$ in $\varphi(x)$ is in the scope of 
some modality. We say a formula $\varphi$ is \intro{guarded} if each bound variable in $\varphi$ is guarded. 
The following fact is classical, but for the sake of completion we repeat it here,
\begin{lemma}
 Every formula is equivalent to a formula which is furthermore guarded.
\label{guarded}
 \end{lemma}
\begin{proof}
 Proof is by induction on the structure of the formula. 
The atomic, boolean and modal cases are straightforward. The non-trivial case is when the formula 
is of the form $\lambda x. \varphi(x)$. Assume $\lambda x. \varphi(x)$ is unguarded  and $\varphi(x)$ is
guarded. We can furthermore assume that all unguarded occurrences of $x$ is outside of any subformula
$\theta y. \psi(x,y)$ of $\varphi(x)$, otherwise in $\varphi(x)$ we substitute   
for $\theta y. \psi(x,y)$ the equivalent formula $\psi(x,\theta y.\psi(x,y))$ which yields  
the desired form. Next we write $\varphi(x)$ is conjunctive normal form to obtain a formula 
of the form
$$\lambda x. ( x \vee \alpha(x)) \wedge \beta(x) ,$$ 
where $\alpha(x)$ and $\beta(x)$ are guarded. It is left to the reader to check that
$$\mu x. ( x \vee \alpha(x)) \wedge \beta(x) \equiv \mu x. \alpha(x) \wedge \beta(x) \,,$$
and
$$\nu x. ( x \vee \alpha(x)) \wedge \beta(x) \equiv \nu x. \beta(x) \,.$$
\end{proof}

We will be using the modalities defined above liberally.
The zeroary modalities $\S$ and $\P$ are used to capture $\fotwo$. They are definable in $\mu$-calculus only using unary modalities and the $\nu$
operator.
\begin{proposition}
\label{proposition:marking}
 The modalities $\S$ and $\P$ are definable in $\nu$-fragment in terms of the unary modalities.
\end{proposition}
\begin{proof} 
We claim that $\S \equiv \nu x. \nextg \prevc x$ and $\P \equiv \prevg \S$. 

Regarding the proof of the claim we want to remark that the proof exploits the same idea used in 
\cite{BjorklundS10} to prove that there is a data automaton which guesses and verifies the marked string 
projection of a data word.

Fix a data word $w$. It is clear that if $i \sim i+1$ then $w,i \models \nu x. \nextg \prevc x$. It only remains to show that
If $i \not \sim i+1$ then $w,i \not \models \nu x. \nextg \prevc x$. Consider the sequence of positions $i_0=i, i_1, \ldots$ such that 
for every $j\in \N$ it is the case that $i_j+1=i_{j+1}\succclass$ (or in other words $i_{j+1}$ is the class predecessor of the successor of $i_j$). 
We claim that this sequence is finite. From this claim, it follows that $w,i \not \models \nu x. \nextg \prevc x$ since there is no infinite path from
$i$.
It is enough to show that for every $j\in$ it is the case that $i_{j+1} < i_j$ since the data word is of finite length.
We prove this claim using induction. The base case of $i_1 < i$ follows from the assumption that $i \not \sim i+1$ (since, either $i+1$ does not have 
a class predecessor or it is strictly below $i$). For the inductive step assume that the claim is proved for $i_0, \ldots, i_{j-1}, i_j$. 
Consider $i_j$ and $i_{j+1}$. Since $i_j+1 = i_{j+1}\succclass$ it is clear that $i_{j+1} \leq i_j$. It remains to show that
$i_{j+1} \neq i_j$. Assume on the contrary $i_{j+1} = i_j$. This means that $i_j +1 = i_{j+1}\succclass = i_{j} \succclass$. It follows that
$i_j+1 = i_j\succclass = i_{j-1}+1$, since successor function $+1$ is an injection, we deduce that $i_j = i_{j-1}$. But by induction hypothesis,
$i_j < i_{j-1}$ which is a contradiction. Therefore the inductive step $i_{j+1} < i_j$ is proved. 
From our claim it follows that the sequence strictly decreases. Since the set of positions is well-founded, the sequence is finite. 
Therefore the formula $\nu x. \nextg \prevc x$ is not true at $i$. This proves our claim that $\S \equiv \nu x. \nextg \prevc x$.

Let us observe that 
$\S$ is in the $\nu$-fragment and so is $\P$ (hence $\neg \S$, $\neg \P$ are in the $\mu$-fragment). By definition the 
formula $\nu x. \nextg \prevc x$ is not in BR, however we do not know if there is a formula which is equivalent to $\S$ which is in 
BR (See Section \ref{section:bounded-reversal}). Readers who are familiar with register automata or data monoids will immediately recognize that the
formula $\S$ and its negation both are
recognizable by a data monoid (in fact this is one of the examples provided in \cite{Bojanczyk11data}) and hence by a deterministic 
one register automata. We conjecture that $\neg \S$ is not in $\nu$-fragment, which will separate our largest decidable fragment and data monoids.

The idea used in the proof of the above proposition can be extended easily to define similar zeroary modalities which indicates how a position and its
$k$-th successor compares with respect to $\sim$. For instance consider the modality $\S_2$ which says that 
that the successor of the successor of a position $i$ is the class successor of $i$. Formally $w,i \models \S_2$ if 
$i \succclass= (i+1)+1$. Let $\mathsf{Even}$ denote the $\mu$-calculus formula which is true at all even positions. Then using ideas similar to that
of the above proof 
it can be shown that $\S_2$ is also definable in $\mu$-calculus in the following way,
\begin{align*}
\S_2 \defeq \( \nu x. \mathsf{Even} \wedge \nextg \nextg \prevc (\mathsf{Even} \wedge x )\) & \\
        &\hspace{-2.5cm}\vee \( \nu x. \neg \mathsf{Even}  \wedge \nextg \nextg \prevc (\neg \mathsf{Even} \wedge x )\)\ .\\
\end{align*}
We note that similarly the modality $\S_n$ can be defined which says that the $n$-th successor of a position is its class successor. 

Consider the modality $\S_{k,n}$ which is true at a position $i$ if the $n$-th successor of $i$ is the $k$-th class successor of $i$. Such a formula
can be written as disjunction of formulas using unary modalities and $\S_1,\ldots, \S_n$. This shows that the modality $\S_{k,n}$ is also expressible
in $\mu$-calculus.

Finally let us remark that all these formulas are recognizable by register automata and also by data automata, 
since register automata are subsumed by data automata \cite{BjorklundS10}. Therefore adding these formulas to our language does not affect the 
decidability of the $\nu$-fragment.

\end{proof} 

But we do not know if $\neg \S$ and $\neg \P$ (obviously definable using $\mu$ operator) 
are definable using $\nu$ operator only (we conjecture negatively). However since these formulas are definable using a data automaton 
(which is our tool for showing decidability) adding them to our language does not affect any of the decidability results.

\subsection{The $\mu$-fragment}

\mlabel{subsection:undecidability-mu}
We consider in this section the \intro{$\mu$-fragment} of $\mu$-calculus, which is the restriction
to the use of least-fixpoints $\mu$ only. The main result is to show the undecidability of its satisfiability.

Consider a data word that uses, say, letters $a,b,c$, and such that
the relation $\sim$ between positions
is a bijection between $a$-labeled positions and $b$-labeled positions.
It is easy to write a $\mu$-calculus formula that checks this property.
However, this is not yet sufficient for our purpose.
We need the following lemma.
\begin{lemma}\mlabel{lemma:monotonic}
The exists a formula in the $\mu$-fragment that checks over finite data words the property
that $\sim$ is an increasing bijection between $a$-labeled positions and $b$-labeled positions.
\end{lemma}
\begin{proof}
For the sake of explanations, let us consider a data word $u$, and
let $A$ ({\it resp.} $B$) be the set of $a$-labeled ({\it resp.} $b$-labeled) positions in $u$.
Let $\mathbin{R}$ be $\sim$ restricted to $A\times B$.
We have to provide a formula that holds if $R$ is a monotonic bijection between $A$ and $B$.
It is easy to write a formula of the $\mu$-fragment that holds if and only if $R$ is a bijection between $A$ and $B$.
We assume this is the case from now.

Consider now the binary relation $S\subseteq A^2$ such that
$x\mathbin{S}z$ if $x \mathbin{R} x'< y' \mathbin{R^{-1}} y<z$. 
An element $x\in A$ such that $x\mathbin{S} x$ is called a \intro{small witness}. 
Note first that the the existence of a small witness means that 
there exists $x>y$ and $x'<y'$ such that $x\mathbin{R} x'$ and $y\mathbin{R}y'$.
Hence, there exists a small witness if and only if $R$ is not increasing.
Unfortunately, we are not able to directly detect the existence of a small witness using a $\mu$-formula.
Instead, we will search for `big witnesses'.
A \intro{big witness} is a sequence $x_1,x_2,\dots$ of elements of $A$ such that
		$$x_1 \mathbin{S} x_2 \mathbin{S}\dots$$

We claim $(\star)$ that there exists a small witness if and only if there exists a big witness.
Of course, if there is a small witness, there is a big one. Assume now that there exists a big witness $x_1,\dots$
Since the $x_i$'s range over a finite domain, there exists $i$ such that $x_{i+1}\leq x_i$.
Thus, $x_i\mathbin{S} x_{i+1}\leq x_i$ and hence $x_i\mathbin{S}x_i$.
we have found a small witness. 

One easily verifies now that the $\mu$-formula
$$\Fg\nu x.a\wedge \Fc\Pc (b\wedge \nextg\Fg (b\wedge \Fc\Pc(a\wedge \nextg \Fg x))))\ $$
expresses the existence of a big witness. Thus the non-existence of a big witness, hence of a small witness, 
hence the non increasing nature of $R$ is definable by a $\mu$-formula. {\it A priori}, this formula
is a formula that uses both $\mu$- and $\nu$-fixpoints since the modalities $\Fc$ and $\Fg$ are in fact syntactic sugar
for formulas of the $\mu$-fragment. However, it is easy to check that, over \emph{finite data words}, $\Fg(\varphi)$
is equivalent to $\nu x. \varphi\wedge\nextg x$ (the difference between least and greatest fixpoint does not exist when the
fixpoints are reached within a finite number of steps). Thus, the above formula can be expressed in the $\nu$-fragment,
and hence its complement in the $\mu$-fragment.
\end{proof}

Using this lemma we reduce the Post's correspondence problem to the satisfiability problem of the logic giving us,

\begin{theorem}\label{theorem:undecidability}
Satisfiability of the $\mu$-fragment over data words is undecidable. 
\end{theorem}
\begin{proof} 
The proof is by reduction from the Post's Correspondence Problem (PCP).
An instance $I$ of PCP is a finite set of tuples  $I=\{(u_1, v_1), \ldots, (u_k,v_k)\mid u_j,v_j \in \Sigma^+\}$.
A solution to $I$ is a sequence $i_0 \ldots i_n\in [k]^{+}$ such that
 $u_{i_0}\ldots u_{i_n}=v_{i_0}\ldots v_{i_n}$. It is well known that the problem of determining if an instance
 of the PCP has a solution is undecidable.

Given an instance $I$ of the PCP,
we construct a formula in the $\mu$-fragment that is satisfiable if and only if $I$ has a solution.
For this, we encode the solution of $I$ as a data word $u$ over the alphabet
 $\Sigma \uplus \{a,b\}$ (where $a,b$ are assumed not present in $\Sigma$). Intuitively, $u$ is
 $u_{i_0}\ldots u_{i_n}$ in which are inserted letters $a$ and $b$ letters in order to describe the decomposition in $u_{i_0},\ldots,u_{i_n}$
(using $a$'s) and in $v_{i_0},\ldots,v_{i_n}$ (using $b$'s). The data values are required to induce an increasing bijection between $a$-labeled and
$b$-labeled positions in order to be able to check the correctness of the solution.
Formally, a data word $u$ \intro{encodes} the solution $i_0\ldots i_n$ to $I$ if:
\begin{itemize}
\item the word has length at least 4, starts with letters $ab$ and ends with $ab$, and
\item $\sim$ induces an increasing bijection between $a$-labeled positions and $b$-labeled positions.
	 Let $x_0<\dots<x_n$ be the $a$-labeled positions and $y_0<\dots< y_n$ be the $b$-labeled positions.
 \item Then for all $\ell=1\dots n$, the word obtained as the string projection of $u$ restricted to
	the positions in $(x_{\ell},x_{\ell+1})$ ({\it resp.} $(y_{\ell},y_{\ell+1})$)to which $b$-letters ({\it resp.} $a$-letters) are removed is
$u_{i_\ell}$ ({\it resp.} $v_{i_\ell}$).
 \end{itemize}
 It is easy, from a solution to construct a data word that encodes it.

 Hence, in order to guess a solution to $I$, it is sufficient to guess a data word over the alphabet $\Sigma\cup\{a,b\}$ such that ($\dagger$):
\begin{itemize}
\item the word has length at least 4, starts with letters $ab$ and ends with $ab$, and
\item $\sim$ induces an increasing bijection between $a$-labeled positions and $b$-labeled positions, and there is at least one occurrence of $a$;
 \item for all occurrences $x$ of an $a$-letter, but the last one, there exists $i\in[k]$ such that:
 		\begin{itemize}
		\item the string projection of $u$ starting at position $x$ belongs to $K_i=\{w~:~\overline w^b \in au_ia(\Sigma\cup a)^*\}$ where
$\overline w^b$ is the word $w$ with letter $b$ removed, and
		\item the string projection of $u$ starting at position $R(x)$ belongs to $L_i=\{w~:~\overline w^a \in bv_ib(\Sigma\cup b)^*\}$ where
$\overline w^a$ is the word $w$ with letter $a$ removed.
		\end{itemize}
 \end{itemize}
 Quite naturally, if a data word encodes a solution to $I$ then it satisfies ($\dagger$). Conversely, if a data word satisfies ($\dagger$),
 then there exists a solution to $I$ that it encodes.
 
 Thus, it is sufficient for us to write a formula of the $\mu$-fragment for ($\dagger$), which is easy using Lemma~\ref{lemma:monotonic}
 for the second item, and the fact that the languages $K_i$ and $L_i$ are regular, thus definable by a formula of the $\mu$-fragment.
\end{proof} 

The above theorem extends to $\omega$-words.
\begin{corollary}\label{corollary:undecidability-omega}
Satisfiability of the $\mu$-fragment over data $\omega$-words is undecidable. 
\end{corollary}
\begin{proof}
Consider a formula $\varphi$ of the $\mu$-fragment, our goal is to construct a formula $\varphi^\sharp$ such that $\varphi$
is satisfiable over data words if and only if $\varphi^\sharp$ is satisfiable over $\omega$-data words. In combination with
Theorem~\ref{theorem:undecidability}, this proves the statement.

The formula $\varphi^\sharp$ (for $\sharp$ a new fresh symbol) defines the data $\omega$-words $w$
such that:
\begin{itemize}
\item $w$ contains at least one occurrence of the letter $\sharp$,
\item the data $\omega$-word $w$ restricted to the positions that are to the left of all $\sharp$-occurrences satisfy $\varphi$.
\end{itemize}
Of course, if we can write such a formula, then it is satisfiable over data $\omega$-words if and only if $\varphi$ is satisfiable over data words. It
is also clear that the first item is definable in the $\mu$-fragment. Thus, we just have to turn $\varphi$ into a formula that is sensitive only to
the part of the word left of all $\sharp$'s. This is exactly the classical technique of relativization. Remark first that the property `being at the
left of all $\sharp$' is definable in the $\mu$-fragment. Let $\psi$ be such a formula. In our case,
relativizing $\varphi$ to $\psi$ consists in replacing syntactically every subformula of the form $\mathsf{M}(\gamma)$ for some modality
$\mathsf M\in\{\nextc,\nextg,\prevc,\prevg\}$ by $\mathsf{M}(\gamma\wedge\psi)$, $\lastg$ by $\nextg\sharp$ and $\lastc$
by $\lastc\vee\nextc\Sg\sharp$. The result is a formulas that holds over a
word if and only if $\varphi$ holds on the input restricted to its longest $\sharp$-free prefix.
\end{proof}

\subsection{The $\nu$-fragment}
 
Fortunately, the $\nu$-fragment is decidable.
We show that for every formula in the $\nu$-fragment there is an equivalent data automaton, which immediately yields the decidability
of the fragment as well.

\begin{theorem}\mlabel{theorem:nu-to-data}
For every formula $\varphi$ in the $\nu$-fragment there is an effectively constructed Data $\omega$-automaton
$\A_\varphi=(B,\Sigma',C,C_\omega)$ such that $\varphi$ and $\A_\varphi$ define the same data $\omega$-language. Moreover the data automaton
$(B,\Sigma',C)$  and $\varphi$ define the same data language.
\end{theorem}
\newcommand{\ST}{\mathit{ST}}
\begin{proof} 
It is a general fact that the $\nu$-fragment of $\mu$-calculus over a set of modalities that are definable in $\fotwo$ can be
defined in $\emsotwo$ using the standard translation. This fact along with the theorem \ref{theorem:dataemso2}
implies that $\nu$-fragment is subsumed by data automata. In the following we give the standard construction
for the $\nu$-fragment which will be used elsewhere in the paper.

We need the following definitions. Let $\mathrm{Prop}(\varphi)$ be the set of all propositional variables used in $\varphi$,
and let $\mathrm{Sub}(\varphi)$ be the set of all subformulas of $\varphi$.
\begin{definition}
\label{definition:closure}
The {\em closure} $\fl\(\varphi\)$ of $\varphi$ is the smallest set such that,
\begin{enumerate}
 \item $\mathrm{Prop}(\varphi) \cup \{ \varphi, \S, \P, \firstc, \firstg, \lastc, \lastg\}$ and their negations belong to $\fl(\varphi)$,
 \item If $\psi \in \fl(\varphi)$ then $\neg \psi$ (negation is pushed to the literals) belongs to $\fl(\varphi)$,
 \item If $\varphi_1 \wedge \varphi_2 \in \fl(\varphi)$ or $\varphi_1 \vee \varphi_2 \in \fl(\varphi)$ then $\varphi_1 \in \fl(\varphi)$ and
$\varphi_2 \in \fl(\varphi)$,
 \item If one of $\nextc \varphi_1, \nextg \varphi_1, \prevc \varphi_1, \prevg \varphi_1$ is in $\fl(\varphi)$, then $\varphi_1 \in \fl(\varphi)$,
 \item If $\nu x. \varphi_1(x) \in \fl(\varphi)$ then $\varphi_1 (\nu x. \varphi_1(x)) \in \fl(\varphi)$.
 \item If $\mu x. \varphi_1(x) \in \fl(\varphi)$ then $\varphi_1 (\mu x. \varphi_1(x)) \in \fl(\varphi)$.
 \end{enumerate}
 \end{definition}
 \begin{definition}
 \label{definition:atom}
 An {\em atom} $A$ is a subset of $\fl(\varphi)$ that satisfies the following properties:
\begin{enumerate}
 \item For all $\psi \in \fl(\varphi)$, $\psi \in A$ iff $\neg \psi \not \in A$,
 \item For all $\varphi_1 \vee \varphi_2 \in \fl(\varphi)$,  $\varphi_1 \vee \varphi_2 \in A$ iff $\varphi_1 \in
A$ or $\varphi_2 \in A$,
 \item For all $\nu x. \varphi_1(x) \in \fl(\varphi)$,  $\nu x. \varphi_1(x) \in A$ iff $\varphi_1(\nu x.
\varphi_1(x)) \in A$. 
\end{enumerate}
\end{definition}
Now we describe how the data $\omega$-automaton $\A_{\varphi}=(B, \Sigma', C, C_\omega)$ works on a given data $\omega$-word $w$. The internal
alphabet $\Sigma'$ is precisely the set
of all 
 atoms in $\fl(\varphi)$. The automaton $B$ while reading the marked string projection of $w$ labels each position with an atom $A_i$
and outputs it.
It 
also verifies that
\begin{enumerate}[(i)]
 \item $\firstg \in A_i$ iff $i$ is the first position and $\lastg \in A_i$ iff $i$ is the last position,
 \item $p \in A_i$ iff the label at position $i$ is $p$, 
 \item let $\tp{i} = (p,s)$ then $\S \in A_i$ iff the marking $s=\S$, similarly, $\P \in A_i$ iff the marking $p$ is $\P$,
 \item $\nextg \varphi_1 \in A_i$ iff $\varphi_1 \in A_{i+1}$.
 \item $\prevg \varphi_1 \in A_i$ iff $\varphi_1 \in A_{i-1}$,
 \item $A_1$ contains $\varphi$.
\end{enumerate}

The class automata $C$ and $C_\omega$ running over a class verifies that,
\begin{enumerate}[(a)]
\item $\firstc \in A_i$ iff $i$ is the first position of a class and $\lastc \in A_i$ iff $i$ is the last
position of a class,
\item $\nextc \varphi_1 \in A_i$ iff $\varphi_1 \in A_{i \plusc 1}$,
\item $\prevc \varphi_1 \in A_i$ iff $\varphi_1 \in A_{i \minusc 1}$.
\end{enumerate}

To show the correctness of the construction assume that $w \in L(\varphi)$ and consider the run of $B$ in which the word $w$ is labelled with the
atoms
$A_i$ such that formulas in $A_i$ hold at position $i$. It follows from definitions that both $B$, $C$ and $C_\omega$ have successful runs on this
particular
transduction and hence the word is accepted.

For the other direction we need to show that ($\star$) if $\A_{\varphi}$ has a successful run on $w$ then $w \in L(\varphi)$.
Observe that if $\A_{\varphi}$ has a successful run on $w$ then there is an annotation $A_1, A_2,\ldots,$ of it which satisfy 
the conditions (i--vi) and (a--c). To
prove ($\star$) we prove the stronger claim that   
{\em For every formula $\varphi$ in the $\nu$-fragment and for every data word $w$ and for every sequence $A_i$ of atoms in $\fl(\varphi)$ satisfying 
conditions (i--vi) and
(a--c) and for every $\psi \in \fl(\varphi) \cap \mathrm{sub}(\varphi)$, if $\psi\in A_i$ then
$w,i \models \psi$}. Obviously this claim in conjunction with condition (vi) implies ($\star$). Proof is by induction on
the structure of the formula. For propositions,
their negations, and zeroary modalities the claim is guaranteed by the conditions (i--iii) and
(a). For the case of boolean operators 
and unary modalities,
we use induction hypothesis and conditions (iv-v) and (b-c). 
The only remaining case is when $\psi$ is of the form $\nu x. \chi(x)$. Consider the data word $w[\ell(x):=\{ i \mid
\psi \in A_i\}]$. Let $A_1',A_2',\ldots $ be the sequence of atoms in $\fl(\chi(x))$ (considering $x$ as a propositional variable) uniquely defined as
$A_i'= \{ \phi[\nu x.\chi(x)/x] \mid \phi \in A_i\} \cap \fl(\chi(x))$. One can easily verify that
$A_1',A_2',\ldots$ satisfy the conditions (i--vi) and
(a--c) on the data word $w[\ell(x):=\{ i \mid
\psi \in A_i\}]$. Hence by induction hypothesis $w[\ell(x):=\{ i \mid
\psi \in A_i\}], i \models \chi(x)$. Therefore the set  
$\{ i \mid
\psi \in A_i\}$ is a post-fixpoint of the function $\chi(x)$ on $w$. Since the greatest
fix point subsumes any post-fixpoint we conclude that for any position $i$ such that $\nu x. \chi(x) \in A_i$
it is the case that $w,i\models \nu x. \chi(x)$.
\end{proof}

We dont know if the containment of $\nu$-fragment in DA is strict. The decidability of the $\nu$-fragment follows from the above theorem. We also note
that the $\nu$-fragment is not closed effectively under complement
since it is decidable
while its complement is not decidable. In fact, building on the formulas used for undecidability of the $\mu$-fragment, we can prove that it is not
closed under complement, 
even non-effectively. Let us finally note that the $\nu$-fragment extended with the zeroary predicates discussed in the previous section is also 
decidable by translation to data automata.

%% file: bounded-reversal.tex
% !TEX root =  main.tex

\section{The bounded reversal and bounded mode alternation fragments}
\mlabel{section:bounded-reversal}
In this section we introduce the main fragments discussed in the paper, namely Bounded Reversal (BR) and Bounded Mode Alternation (BMA).
We begin by presenting the $\comp$ hierarchy, which is the logical counterpart to cascade of automata, we then introduce the
BR and BMA fragments.

\subsection{Composition and the BR and BMA logics}
\mlabel{comp}

Before delving into the technical details let us outline the intuition behind each of the fragments.
Each modality in the $\mu$-calculus goes either left ($\prevg,\prevc$)
or right ($\nextg,\nextc$) to evaluate the argument formula. 
A formula is in the BR fragment if the number of times
the formula switches between the ``left'' and ``right'' directions is bounded. 
Just like every modality in our logic has a direction, it has a mode.
Each modality in the $\mu$-calculus is either a class modality ($\nextc,\prevc$)
or a global modality ($\nextg,\prevg$). A formula is in the BMA fragment if the number of times the formula
switches between the ``class'' mode and ``global'' mode is bounded.
The formal way to describe these fragments is as composition of formulas that are purely ``left'' or purely ``right''
(in the BR case), or purely ``global'' or purely ``class'' (in the BMA case). This is done
using the $\comp$-operator from $\mu$-calculus.
\begin{definition}%[The composition operation $\comp$]
Let $\Psi$ be a set of $\mu$-calculus formulas.  Define the sets 
\begin{itemize}
\item $\comp^0(\Psi)=\emptyset$,
\item $\comp^{i+1}(\Psi)=\{\psi(\varphi_1, \ldots, \varphi_n)~|~\psi(x_1,\ldots, x_n)\in\Psi,~\varphi_1,\ldots, \varphi_n\in \comp^i(\Psi)\}$ where
the substitution follows the usual condition that none of the free variables of 
$\varphi_1,\ldots, \varphi_n$ get bound in  $\psi(\varphi_1, \ldots, \varphi_n)$. 
\end{itemize}
The set of formulas $\comp(\Psi)$ is defined as $\comp(\Psi) = \bigcup_{i \in \N} \comp^i(\Psi)$. For a formula $\psi\in \comp(\Psi)$ 
we define the {\em $\comp$-height of $\psi$ in $\comp(\Psi)$} as the least $i$ such that $\psi \in  \comp^i(\Psi)$.
\end{definition}

Next we formally define BR and BMA. 
If $M$ is a set of modalities, then $\mathsf{Formulas}(M)$ is defined as the subset of $\mu$-calculus which uses only 
the modalities $M$ (apart from the zeroary modalities). 
\begin{definition}[BR and BMA] \mlabel{definition:BR}\mlabel{definition:BMA}
Let $M_{\mathtt{X}} = \{\nextc, \nextg\}$, $M_{\mathtt{Y}}=\{\prevc, \prevg\}$,  
$M_{\mathit{g}}=\{\nextg, \prevg\}$ and $M_{\mathit{c}}=\{\nextc, \prevc\}$.

The \intro{BR fragment} of $\mu$-calculus is the set of formulas  
$\comp\(\mathsf{Formulas}\(M_{\mathtt X}\)\cup \mathsf{Formulas}\(M_{\mathtt{Y}}\)\))$.

The  \intro{BMA fragment} of $\mu$-calculus is the set of formulas 
$\comp\(\mathsf{Formulas}\(M_{\mathit{g}}\) \cup \mathsf{Formulas}\(M_{\mathit{c}}\)\)$.
\end{definition}
\begin{example} \mlabel{brexample}
Define
\begin{align*}
\varphi_1&=\nu x.(\dnextc x \vee \nextg \mu y. ( q \wedge \dprevc y )), \varphi_2 = \nu x. \( \nextc\last \vee \nextc\prevg x\),\\
\varphi_3&=\mu x. ((\nu y.\, q \vee \nextc y) \vee \nextg x \vee \prevg x), \varphi_4 = \mu x.(\nextc\nextg x\vee p).
\end{align*}
The formula $\varphi_1$ is in BR (comp-height 2) and in BMA (comp-height 3). The formula $\varphi_2$ is neither in BR nor in BMA.
The formula $\varphi_3$ is in BMA (comp-height 2) but not in BR. The formula $\varphi_4$ is in BR (comp-height 1) but not in BMA.
\end{example}
%\begin{example} 
%Define the language $\mathsf{Path_k}$ as the set of all data words $w$ such that there 
%is a path from the first position to the last position
%using global successor and class successor edges with the restriction that the path consists of at most $k$ global successor edges. 
%This language is described by the formula 
%$\mathit{path}_k$ where the formula $\mathit{path}_i$ is defined inductively as,   
%$\mathit{path}_0 = \mu x_1.\( \nextc x_1 \vee \last\)$ and 
%$\mathit{path}_{i+1} = \mu x_{i+1}.\( \nextc x_{i+1} \vee \dnextg \mathit{path}_i \vee \mathit{path}_i\)$
%\end{example}
\newcommand{\bridge}{\mathsf{Bridge}}
\begin{example}
\mlabel{Example:Bridge}
Define the language $\bridge_k$ as the set of all data words such that,
by applying global successor, followed by class successor, \dots ($k$-times),
one reaches a position labeled with letter $a$.
This language is described by the formula, 
$$\overbrace{\nextg\nextc\dots \nextg\nextc}^{\text{$k$-times}}a\ .$$   
It is BR (of comp-height 1) and in BMA (of comp-height $2k$).
The language $\bridge$ is the union of all $\bridge_k$, and can be described
by the formula $\mu x.(\nextg\nextc x\vee a)$. It is BR (of comp-height 1) but not in BMA.
\end{example}

\begin{theorem}[BMA $\subseteq$ BR]
\label{theorem:bma-to-br}
For every formula $\varphi$ in BMA of $\comp$-height $k$ there is an equivalent (over data words and data $\omega$-words) 
formula $\varphi'$ in BR of $\comp$-height $k+1$.
 \end{theorem} 
 \begin{proof} 
We prove the following claim by induction, {\it for every formula of $\varphi$ in BMA of $\comp$-height $k$ there is an 
there is an equivalent (over data words and data $\omega$-words) formula $\varphi'$ which is a boolean combination of 
formulas in BR of $\comp$-height $k$.} Note that since a boolean combination of BR formulas of $\comp$-height $k$
has $\comp$-height $k+1$ the theorem follows.

For the base case let $\varphi$ be in $\mathsf{Formulas}\(M_{\mathit{g}}\) \cup \mathsf{Formulas}\(M_{\mathit{c}}\)$ 
(of $\comp$-height $1$). Consider the case when $\varphi$ is in $\mathsf{Formulas}\(M_{\mathit{g}}\)$. 
Let $w$ be a data word (\resp{}. data $\omega$-word) and $i$ be a position in $w$,
The idea is to translate $\varphi$ into an equivalent finite state (\resp{} B\"{u}chi) 
automaton and re-encode it as a boolean combination of $\mathsf{Formulas}\(M_{\mathtt{X}}\) \cup \mathsf{Formulas}\(M_{\mathtt{Y}}\)$. 
One can think of 
 $\varphi$ as a formula evaluated over a word ($\omega$-word) 
 $w$ over the alphabet $P=2^{\mathit{Prop}(\varphi)} \times \M$. 
 Utilizing the correspondence between $\mu$-calculus and finite state
 (\resp{} B\"{u}chi) automata, there is a finite state (\resp{} B\"{u}chi) automaton 
 $A_\varphi=\(Q, P, \Delta, q_0, F\)$ with the set of states
$Q$, the set of transitions 
 $\Delta \subseteq Q \times P \times Q$, 
 the initial state $q_0$ and the set of final states (\resp{} B\"uchi states) $F$, 
 equivalent to $\varphi$ in the following sense. There is a state
$q\in Q$ such that
if 
 $A_{\varphi}$ has a successful run $\rho=q_0 q_1 \ldots q_n$(\resp{} $\rho=q_0 q_1 \ldots$)  then for all positions $i$, it is the case that 
 $w,i \models \varphi$ if and only if $q_i = q$.
 Therefore to verify that $w,i \in \varphi$ it is enough to check that 
 (1) the automaton $A_{\varphi}$ has a run
 starting in the state $q_0$ ending in state $q$ on the prefix $w[1:i]$ (2) $A_{\varphi}$ has a successful run
 starting in the state $q$ on the suffix $w[i+1:n]$ (\resp{} $w[i+1:\infty]$). 
 We can encode condition (1) using a $\mu$-calculus formula using only
 the modality $\prevg$ and condition (2) using a formula using only the modality $\nextg$. Thus $\varphi$ is equivalent to a boolean
combination of formulas in $\mathsf{Formulas}\(M_{\mathtt{X}}\) \cup \mathsf{Formulas}\(M_{\mathtt{Y}}\)$.
When $\varphi$ is in $\mathsf{Formulas}\(M_{\mathit{c}}\)$ the construction is similar except that while encoding the run of the automaton
$A_{\varphi}$ we use the 
 modalities $\prevc$ and $\nextc$.
 
 For the inductive step, let $\varphi = \psi(\varphi_1,\ldots,\varphi_k)$ be a BMA formula of $\comp$-height $k+1$ 
 where $\psi(x_1,\ldots,x_k)\in
\mathsf{Formulas}\(M_{\mathit{g}}\) \cup \mathsf{Formulas}\(M_{\mathit{c}}\)$ and $\varphi_1,\ldots, \varphi_k$ are BMA formulas of $\comp$-height
$k$. Using induction hypothesis we obtain $\varphi_1',\ldots, \varphi_k'$ which are boolean combinations of BR formulas
of $\comp$-height $k$ 
and are equivalent to
$\varphi_1,\ldots, \varphi_k$ respectively. Repeating the previous argument we also obtain
$\psi'(x_1,\ldots,x_k) \in \mathsf{Bool}(\mathsf{Formulas}\(M_{\mathtt{X}}\) \cup \mathsf{Formulas}\(M_{\mathtt{Y}}\))$ equivalent
to $\psi(x_1,\ldots, x_k)$. To conclude observe that $\psi'(\varphi_1',\ldots,\varphi_k')$ is a boolean combination
of BR formulas of $\comp$-height at most $k+1$.
 \end{proof}

Next we show that BR is subsumed by the $\nu$-fragment over data words. The result extends to data $\omega$-words partially.

\begin{lemma}
Let $\varphi(x,\bar{y})$ be a formula such that the only unary modalities it uses are $\prevg,\prevc$
and furthermore any free occurrence of $x$ appears in the scope of at least $k$ nested modalities. Then for any data 
word (\resp{} data $\omega$-word) $w$ and valuation
$S_1, \ldots, S_l$ of $\bar y=y_1,\dots,y_l$, and $S$ of $x$, and for all $i<k$,
\begin{align*}
w[\ell(\bar {y}):=\bar{S}, \ell(x) = S],i &\models \varphi \\ 
&\Leftrightarrow w[\ell(\bar {y}):=\bar{S}, \ell(x) =
\emptyset], i \models \varphi\ . 
\end{align*}
\label{lemmaforBR-Y}
\end{lemma}
\begin{proof}
Without loss of generality assume that $x$ is not a bound variable in $\varphi(x,\bar{y})$ (otherwise rename the occurrences of $x$).
We proceed by an induction on the pair $(k,i)$ ordered lexicographically (for all $i\geq k$ the claim holds trivially); For the base case when $k=1$,
the claim is vacuously true. For the inductive step
assume the claim is true for pairs $(k',i')$ where $k'<k$ or, $k'=k$ and $i'<i$. 
Let $\varphi(x,\bar{y})$ be a formula in which $x$ appears with in the scope of $k+1$ nested modalities.
We do an induction on the structure of the formula. 
Let $\varphi(x,\bar{y})$ is of the form $\mathtt{M}\psi(x,\bar{y})$ where $\mathtt{M} \in
\{\prevg,\prevc\}$. We do a case analysis on $\mathtt{M}$.
Assume $\mathtt{M}$ is $\prevg$ (the case when $\mathtt{M}$ is $\prevc$ being analogous) then
\begin{align*}
w[\ell(\bar {y})&:=\bar{S}, \ell(x) = S],i \models \mathtt{M}\psi(x,\bar{y})\\ 
                &\Leftrightarrow w[\ell(\bar {y}):=\bar{S}, \ell(x) = S],i-1 \models \psi(x,\bar{y}) \tag{By defn. of $\prevg$}\\
             &\Leftrightarrow w[\ell(\bar {y}):=\bar{S}, \ell(x) = \emptyset],i-1 \models \psi(x,\bar{y}) \tag{$i < k \Rightarrow i-1 < k-1$,
hence by IH}\\
&\Leftrightarrow w[\ell(\bar {y}):=\bar{S}, \ell(x) = \emptyset],i \models \mathtt{M}\psi(x,\bar{y})\\
\end{align*}
The boolean cases are straightforward. Next assume $\varphi(x,\bar{y})$ is of the form $\theta y_i. \psi(x,\bar{y})$ ($\theta \in \{\mu,\nu\}$). We
have to show that  
\begin{align*}
w[\ell(\bar {y})&:=\bar{S}, \ell(x) = S],i \models \theta y_i. \psi(x,\bar{y}) \\
                &\Leftrightarrow w[\ell(\bar {y}):=\bar{S}, \ell(x) =
\emptyset], i \models \theta y_i. \psi(x,\bar{y})\ .\\ 
\end{align*}
By induction hypothesis (on the structure of the formula)
\begin{align*}
w[\ell(\bar {y})&:=\bar{S}, \ell(x) = S],i \models \psi(x,\bar{y}) \\
                &\Leftrightarrow w[\ell(\bar {y}):=\bar{S}, \ell(x) = \emptyset], i \models \psi(x,\bar{y})\ .
\end{align*}
Hence $S_i$ is a pre-fixpoint (\resp{} post-fixpoint) of $\psi(x,\bar{y})$ on $w[\ell(\bar {y}):=\bar{S}, \ell(x) = S]$
if and only if it is a pre-fixpoint (\resp{} post-fixpoint)
of $\psi(x,\bar{y})$ on $w[\ell(\bar {y}):=\bar{S}, \ell(x) = \emptyset]$. Hence the claim is proved by Knaster-Tarski theorem.
This concludes the induction.
\end{proof}

By symmetry the following lemma also holds,
\begin{lemma}
Let $\varphi(x,\bar{y})$ be a formula such that the only unary modalities it uses are $\nextg,\nextc$ and
furthermore any occurrence of $x$ appears in the scope of at least $k$ nested modalities. Then for any data 
word $w$ of length $n$ and valuation
$S_1, \ldots, S_l$ of $\bar y=y_1,\dots,y_l$, and $S$ of $x$, and for all $i>n-k$,
\begin{align*}
w[\ell(\bar {y}):=\bar{S}, \ell(x) = S],i &\models \varphi \\
&\Leftrightarrow w[\ell(\bar {y}):=\bar{S}, \ell(x) =
\emptyset], i \models \varphi\ . 
\end{align*}
\label{lemmaforBR-X}
\end{lemma}

\begin{theorem}
Every BR-formula is equivalent to a formula of the $\nu$-fragment {\em over data words}. 
\label{theorem:br-to-nu}
\end{theorem}
\begin{proof} 
This is done in two steps. The first step is to transform the formula in BR to 
an equivalent one that is furthermore guarded. This is achieved by Lemma \ref{guarded}.
In the second step we turn every subformula of the form %given formula 
%(note that the formula is also a subformula of itself)
$\mu x. \varphi(x,\bar{y})$ into $\nu x. \varphi(x,\bar{y})$. 
We claim that the resulting formula is equivalent to the original one. 
Thanks to Lemma~\ref{guarded}, we only have to prove the correction of the second step,
which amounts to prove that 
{\it (Claim $\star$) given a guarded BR-formula, it is equivalent over all data words to the formula
in which each $\mu$-fixpoint is turned into a $\nu$-fixpoint.}

Observe first that it is sufficient to prove ($\star$) for formulas in $\mathsf{Formulas}\(M_{\mathtt{X}}\)$. Indeed, from this
result, by symmetry, it also holds for formulas in $\mathsf{Formulas}\(M_{\mathtt{Y}}\)$. 
Note now that given formulae $\phi(x),\phi'(x),\psi$ such that $\phi(x)$ and $\phi'(x)$ are equivalent over all data words, then the same holds for the substitutions $\phi(\psi)$ and $\phi'(\psi)$. Since formulas in BR are obtained from formulas in $\mathsf{Formulas}\(M_{\mathtt{X}}\)$ and $\mathsf{Formulas}\(M_{\mathtt{Y}}\)$ via inductive substitution, this implies ($\star$) for all formulas in BR.

Hence, what remains to be shown is that ($\star$) holds for a formula in $\psi\in\mathsf{Formulas}\(M_{\mathtt{X}}\)$.
Observe that by induction on the structure of the formula it is enough to verify that for every guarded formula
$\psi=\mu x. \varphi(x,\bar{y})  \in\mathsf{Formulas}\(M_{\mathtt{X}}\)$ and for every data word $w$ (of length $n$) and valuation
$S_1,\ldots,S_k$ (all of them subsets of $[n]$) of $\bar y=y_1,\dots,y_k$,
$$\sem{\nu x. \varphi(x,\bar{y})}_{w'} \subseteq \sem{\mu x. \varphi(x,\bar{y})}_{w'}$$
where $w'=w[\ell(y_1):=S_1,\ldots,\ell(y_k):=S_k]$, since the other inclusion follows from the fact that the least fixpoint is always included in the
greatest fixpoint. This reduces to showing that
$$w',i\models \nu x. \varphi(x,\bar{y}) \Rightarrow w',i\models \mu x. \varphi(x,\bar{y})$$
This is exhibited by the following calculation,
\begin{align*}
 w',i  \models \nu x. \varphi(x,\bar{y}) &\Leftrightarrow  w',i  \models \varphi(\nu x. \varphi(x,\bar{y}),\bar{y})\tag{By fixpoint
iteration}\\
 &\Leftrightarrow  w',i  \models \varphi^{n+1}(\nu x. \varphi(x,\bar{y}),\bar{y}) \\
 &\Rightarrow w',i  \models \varphi^{n+1}(\bot,\bar{y}) \tag{By Lemma \ref{lemmaforBR-X}}\\ 
 &\Rightarrow w',i  \models \mu x. \varphi(x,\bar{y}) \tag{By Knaster-Tarski theorem} 
\end{align*}
\end{proof} 
From the proof it follows that,
\begin{corollary}
\label{corollary:BR-unique}
 Every guarded BR-formula has a unique fixpoint on every data word.  
\end{corollary}

\begin{theorem}
Over data $\omega$-words, $\mathsf{Formulas}\(M_{\mathtt{Y}}\) \subseteq \nu\mbox{-Fragment}$. It follows that,
over data words and data $\omega$-words,
 $$\comp\(\mathsf{Formulas}\(M_{\mathtt{Y}}\) \cup \nu\mbox{-Fragment}\)=\nu\mbox{-Fragment}\ .$$
\end{theorem}
\begin{proof}
For data words the claim follows from Theorem \ref{theorem:br-to-nu}. For data $\omega$-words the direction
$$\comp\(\mathsf{Formulas}\(M_{\mathtt{Y}}\) \cup \nu\mbox{-Fragment}\) \supseteq \nu\mbox{-Fragment}$$ is clear.
For the other direction, we redo the claim ($\star$) from the proof of Theorem \ref{theorem:br-to-nu} for
$\mathsf{Formulas}\(M_{\mathtt{Y}}\)$ using Lemma \ref{lemmaforBR-Y}.
\end{proof}

 Let us remark that since the class of languages definable by the $\nu$-fragment is not closed under complement while the class of
languages {\em of data words} definable
by 
 BR is closed under complement, it follows that BR is strictly less expressive than the $\nu$-fragment over data words.

\section{Characterizing BMA and BR as cascades of automata}
\label{rue-de-cascade}

In this section we give the characterization of BR and BMA.
It is classical that composition ($\comp$) corresponds to the natural operation of composing sequential transducers.
Given a $\mu$-calculus formula $\varphi$, we can see it as a transducer that reads the input, and labels it 
with one extra bit of information at each position, representing the truth value of the formula at that point.
Under this view, the composition of formulas corresponds to applying the transducers in sequence:
the first transducer reads the input, and adds some extra labelling on it. Then a second transducer reads the
resulting word, and processes it in a similar way, etc...
If we push this view further, we can establish exact correspondences between the class
BR and BMA, and suitable cascades of transducers. Furthermore, the comp-height of the formula matches the
number of transducers involved in the cascade. 

\subsection{Characterizing BMA} 

In this section we characterize BMA in terms of cascades of letter-to-letter 
functional transducers.

We recall that a \intro{functional letter-to-letter transducer} $\A: \Sigma^* \rightarrow \Sigma'^*$ over words is a nondeterministic finite
state letter-to-letter transducer such that every input word
has at most one output word.
Similarly a functional letter-to-letter transducer $\A_\omega: \Sigma^\omega \rightarrow \Sigma'^\omega$ over $\omega$-words is a nondeterministic
finite
state letter-to-letter B\"{u}chi transducer such that every input word
has at most one output word.

\begin{definition}[Global transducer]
\label{definition:global-tranducer}
 A \intro{global transducer $\G$ over data words} 
 with input alphabet $\Sigma\times \M$ and output alphabet $\Sigma'$
 is a functional letter-to-letter transducer which reads the marked string projection $\msp{w}$ of the input data word $w$ and outputs
$\G(\msp{w})$. This defines the unique output data word $w'$ such that $\gdp{w'}=\gdp{w}$ and 
$\gsp{w'}=\G(\msp{w})$. A \intro{global transducer $\G_\omega$ over data $\omega$-words} is defined exactly in the same way except that $\G_\omega$
is a  functional letter-to-letter B\"{u}chi transducer.
\end{definition}
\begin{definition}[Class transducer]
\label{definition:class-tranducer}
A \intro{class transducer $\L$ over data words} 
 with input alphabet $\Sigma\times \M$ and output alphabet $\Sigma'$
 is a functional letter-to-letter transducer which works in the following way. A copy of  the automaton 
 $\L$ reads the marked class projection $\msp{w|_S}$ of the input
data word $w$ for each class $S$ in $w$ and outputs
$\L(\msp{w|_S})$. The unique output data word is defined to be $w'$ such that 
$\gdp{w'}=\gdp{w}$ and 
$\gsp{w'|_S}=\L(\msp{w|_S})$ for each class $S$ in $w$.

 A \intro{class transducer over data $\omega$-words is a pair $(\L, \L_\omega)$} 
where  $\L$ is as before and $\L_\omega$ is a functional letter-to-letter B\"{u}chi transducer.
The working of the automaton is analogous with the addition that on each finite class
the transduction is done by $\L$ and on each infinite class the transduction is done 
by $\L_\omega$.

\end{definition}

\begin{definition}
\label{definition:cascade-bma}
A \intro{cascade of class and global transducers over data words} $\C$ 
is a sequence $\langle \Sigma=\Sigma_0, \A_1, \Sigma_1, \ldots, \Sigma_{n-1}, \A_n, \Sigma_n\rangle$ such that
$\A_1,\ldots, \A_n$ is a sequence of class and global transducers over data words and for each $i$, the transducer $\A_i$ has input alphabet
$\Sigma_{i-1}\times \M$ 
and output alphabet $\Sigma_i$. A \intro{cascade of class and global transducers over data $\omega$-words} $\C$ 
is defined analogously where each $\A_i$ is either a global or a class transducer over data 
$\omega$-words.
We call $\Sigma_0$ the \intro{input alphabet} of $\C$ and $\Sigma_{n}$ the \intro{output alphabet} of $\C$. Also, $n$ is called the \intro{height}
of the cascade. Let $\mathsf{C}$ ({\it resp.} $\mathsf{C}_\omega$) denote the set of all cascades of class and global transducers 
on data words ({\it resp.} data $\omega$-words).
\end{definition}

Given a cascade of class and global transducers $\C$,
a \intro{successful run} of $\C$ on a given data word (\resp{} data $\omega$-word) $w$ is a sequence $w_0=w, \rho_1, w_1,\rho_2, \ldots, w_n, \rho_n$ 
such that $\rho_i$ is a successful run of $\A_i$ 
on $w_{i-1}$ outputing the data word (\resp{} data $\omega$-word) $w_i$. The language accepted by $\C$ is the set of all
data words $w$ on which $\C$ has a successful run.

Observe that cascades are natural analogue of the $\comp$ operator on sets of formulas. Two cascades $\C_1$ and $\C_2$ 
can be composed to form the cascade $\C_1 \circ \C_2$ if the output alphabet of $\C_1$ and input alphabet of $\C_2$ coincide. 

\begin{remark}
 $\mathsf{C}$ and $\mathsf{C}_\omega$ are closed under composition.
\end{remark}

\begin{remark}
\label{remark:bma-union}
 Global (\resp{} class) transducers are closed under product ($C:\Sigma^* \rightarrow \Sigma_1^*\times \Sigma_2^*$ is the product of $A
 :\Sigma^* \rightarrow \Sigma_1^*$ and $B
 :\Sigma^* \rightarrow \Sigma_2^*$ if $C(w)=(A(w),B(w))$). By the previous remark cascades are closed under product. 
 \end{remark}

Next we establish the equivalence between BMA and cascades. We recall the following classical 
results. 

\begin{fact}
\label{fact:mu-Buchi}
Given a $\mu$-calculus formula $\varphi$ over words ({\it resp.} $\omega$-words) there is a 
non-deterministic finite state ({\it resp.} B\"uchi) functional transducer $\A_\varphi$ 
such that given any word ({\it resp.} $\omega$-word) $w$ the automaton $\A_\varphi$ outputs $1$ (\resp{} $0$) exactly
on those positions where $\varphi$ is true (\resp{} false). Moreover $\A_\varphi$ is deterministic if
$\varphi$ uses only the past modalities, and $\A_\varphi$ is co-deterministic if
$\varphi$ uses only the future modalities. Using closure under union of automata 
we can extend this statement to a finite set of formulas.
\end{fact}

\begin{fact}
\label{fact:Buchi-mu}
Given a nondeterministic finite state automaton  ({\it resp.} B\"uchi) $\A$ 
and a transition $\delta$ of $\A$ 
there is $\mu$-calculus formula $\varphi_\delta$ such that for any word ({\it resp.} $\omega$-word)
$w$ and a position $i$ in $w$, $w,i \models \varphi_\delta$ if and only if there is a successful
run $\rho=\delta_1\delta_2 \ldots$ of $\A$ such that $\delta_i=\delta$.
 It follows that given a letter-to-letter transducer $\A:\Sigma^*\rightarrow \Sigma'^*$ 
(\resp{} $\A:\Sigma^\omega\rightarrow \Sigma'^\omega$) and a letter $a\in \Sigma'$ 
there is a formula $\varphi_a$ such that for any word (\resp{} $\omega$-word)
$w$ and a position $i$ in $w$, $w,i \models \varphi_a$ if and only if there is an output
word $a_1a_2\ldots$ of $\A$ such that (\resp{} $\A_\omega$) $a_i=a$. In particular if the transducer is functional
$\varphi_a$ holds if and only if in the unique output word $a_1a_2\ldots$ it is the case that $a_i=a$.
\end{fact}

\begin{proposition} 
For every BMA formula $\varphi$ on data words (\resp{} data $\omega$-words) there is an equivalent cascade 
$\mathcal{C}_\varphi$ in $\mathsf{C}$ (\resp{} in $\mathsf{C}_\omega$) such that the
$\comp$-height of $\varphi$ is exactly the same as the height of the cascade $\C_\varphi$.
\label{proposition:bma-to-cascade}
\end{proposition}
\begin{proof} 
Observe that it is sufficient to prove that ($\star$) {\em for every formula $\varphi$ in $\mathsf{Formulas}\(\mathtt{M_g}\)$ on data words
(\resp{} data $\omega$-words)
there is a 
global transducer $\varphi_\C$ in $\mathsf{C}$ (\resp{} $\mathsf{C}_\omega$) outputting $1$ (\resp{} $0$) exactly at those positions where
$\varphi$ does (\resp{} not) hold}. By Remark \ref{remark:bma-union} the claim holds for a finite set of formulas.
By symmetry a similar claim holds for $\varphi$ in $\mathsf{Formulas}\(\mathtt{M_c}\)$. Finally since $\mathsf{C}$ and 
$\mathsf{C}_\omega$ are closed under composition by induction on the $\comp$-height the proposition follows.
Note that $(\star)$ is guaranteed by Remark \ref{fact:mu-Buchi}.
\end{proof}

\begin{proposition}
For every  cascade 
$\mathcal{C}$ in $\mathsf{C}$ (\resp{} in $\mathsf{C}_\omega$) there is an equivalent BMA-formula $\varphi_\C$ on data 
words (\resp{} data $\omega$-words) such
that the  height of the cascade $\C$ is exactly the same as
 the $\comp$-height of $\varphi_\C$.
\label{proposition:cascade-to-bma}
 \end{proposition}
\begin{proof}
Let $\A$ be a global transducer with output alphabet $\Sigma'$ . From Fact \ref{fact:Buchi-mu}. we obtain that 
for every letter $a \in \Sigma'$, there is a formula $\varphi_a$ in $\mathsf{Formulas}(\mathtt{M_g})$
such that on input $w$ and position $i$,
$w,i \models \varphi_a$ iff for $a_1a_2\ldots = \A(w)$, $a_i=a$. Analogously the similar claim holds for class
transducers. Since BMA is closed under composition by induction on the height of the cascade the claim generalizes to 
cascades of arbitrary height.
\end{proof}

From \ref{proposition:cascade-to-bma} and \ref{proposition:bma-to-cascade} it follows that,
\begin{theorem}
  BMA on data words (\resp{} data $\omega$-words) and $\mathsf{C}$ (\resp{} $\mathsf{C}_\omega$) are equivalent. 
 \end{theorem}

{\bf Sequentializing $\mathsf{C}$ and $\mathsf{C_\omega}$.} Sequentializing cascades is the analogue of determinizing automata 
(it can also be seen as transfering the semantic notion of functionality to a syntactic notion of determinism or co-determinism). A
\intro{left-sequential} ({\it resp.} \intro{right-sequential})
transducer is a transducer which reads the input from left-to-right ({\it resp.} right-to-left) and produces the output
synchronously. On finite words a transducer is left-sequential ({\it resp.} right-sequential) if the automaton obtained by removing 
the output letters is deterministic ({\it resp.} co-deterministic). It is a classical theorem due to Elgot and Mezei \cite{ElgotMezei}
that every rational function on finite words (i.e. one defined by a functional transducer) is defined by the cascade of 
a left-sequential and right-sequential transducer. A similar result holds also for $\omega$-words due to Carton \cite{Carton10}.
In the case of $\omega$-words a left-sequential transducer, as before, is one where the underlying automaton is deterministic,
while the notion of a right sequential transducer is not immediate as the word does not have a maximal position. In this case
one has to use the notion of a {\it prophetic} automaton (%
% 
% A B\"{u}chi automaton $\A$ with set of states $Q$ and alphabet $\Sigma$ 
% is prophetic if $\cup_{q \in Q} L_q = \Sigma^\omega$
% and for all $p,q\in Q$ if $p \neq q$ then $L_p \cap L_q = \emptyset$,
% where $L_q$ denote the language accepted by $\A$ when $q$ 
% is taken as the unique initial state. 
{\em Prophecy} is a strong form of co-determinism. See \cite{Carton10} for more details.)

\begin{definition}[Cascade of sequential transducers] A global ({\it resp.} class) transducer over data words 
is left-sequential if it is deterministic and it is right-sequential if it is is co-deterministic.
A global ({\it resp.} class) transducer $\G$ (\resp{} $\L,\L_\omega$) over data $\omega$-words 
is left-sequential if it is deterministic (\resp{} both $\L,\L_\omega$ are deterministic).
A global transducer $\G_\omega$ (\resp{} class transducer $(\L,\L_\omega)$) is right-sequential if $\G_\omega$ is prophetic (\resp{} if $\L$ is
right-sequential and 
$\L_\omega$ is prophetic). A cascade of sequential transducers is defined in the obvious way.
\end{definition}

\begin{remark}
 Every cascade in $\mathsf{C}$ ({\it resp.} $\mathsf{C}_\omega$) of height $k$ 
 is equivalent to a cascade of sequential transducers of height at most $2k$. 
\end{remark}
\begin{proof}
 Inductively replace each class ({\it resp.} global) transducer with a cascade of left-sequential and right sequential class
 ({\it resp.} global) transducers.
\end{proof}

\begin{remark}[BMA $\subseteq$ DA]
\label{remark:bma-to-da}
We claim that the class of cascades obtained by removing the restriction of functionality from Definitions \ref{definition:global-tranducer}
,\ref{definition:class-tranducer} and \ref{definition:cascade-bma} is equivalent to data automata. It is easy to see that data automata 
belong to this class. For the other direction, it is sufficient to observe that given a cascade 
$\C=\langle \Sigma=\Sigma_0, \A_1, \Sigma_1, \ldots, \Sigma_{n-1}, \A_n, \Sigma_n\rangle$ of ({\em not necessarily functional}) class and global
transducers (without loss of generality assume $n$ is even and even numbered $\A_i$'s work on class projections 
and odd numbered $\A_i$'s work on global projection)
there is a data automaton $(B,\Sigma',C)$ (\resp{} data $\omega$-automaton $\A=(B,\Sigma',C,C_\omega)$) with the intermediate 
alphabet $\Sigma'=\(\Sigma_1\times\Sigma_2\times \ldots\times \Sigma_{n}\)^*$ which works in the following way;
Note that there is an obvious correspondence between words in $\Sigma'^*$ and tuples of words (of identical length) of the form
$(w_1,w_2,\ldots,w_n)$ where $w_i \in \Sigma_i^*$. We implicitly make use of this correspondence below. The transducer $B$ guesses
the words $w_1,w_2,\ldots, w_n$ and outputs it while verifying that on each odd $i$, $A_i$ has a run on $w_{i-1}$ outputting $w_i$.
The class automaton $C$ (\resp{} $C$ and $C_\omega$) verifies that for each even $i$, $A_i$ has a run on $w_{i-1}$ outputting $w_i$.
It is clear that $\C$ has an accepting run on $w$ if and only if $\A$ has an accepting run on $w$. Hence the claim is shown.
It follows that BMA $\subseteq$ DA.
\end{remark}

\subsection{Characterizing BR}
Take note that we treat BR on data words {\em only} below. The results presented do not extend to data $\omega$-words.
First we formally define cascades of class memory transducers, which is then followed by the proof of the equivalence.

The transducers we use are the transducer versions of \intro{class-memory automata} (CMA for short) introduced in
\cite{BjorklundS10}. A class-memory automaton 
is an automaton which reads the data word from left-to-right and at every position the state depends on the current letter, the previous state and
the state the automaton was in when reading the class-predecessor position. Let us remark that it is known that CMA are equivalent to data automata,
while their deterministic variant is strictly weaker \cite{BjorklundS10}. For characterizing BR
we use cascades of deterministic CMA transducers which reads the data word either from left to right and from right to left.

\begin{definition}[Class-memory transducers]
A \intro{deterministic class-memory transducer} (denoted by $\DCMT$) $\A$ is given by a tuple 
$(Q, \Sigma, \Sigma', \Delta, q_0, F_c, F_g)$
where $Q$ is the finite set 
of states, $\Sigma$ is the input alphabet, $\Sigma'$ is the output alphabet, 
$\Delta : Q \times Q\cup\{\top, \bot\} \times \Sigma \times \M  \rightarrow Q \times \Sigma'$ is the transition function, $q_0$
is the initial state,
$F_c$ is the set of class final states and $F_g$ is the set of global final states. 

A \intro{forward} (\resp{} \intro{backward}) deterministic class-memory transducer is a $\DCMT$
which reads its input data
word 
from left-to-right (\resp{} right-to-left).
\end{definition}

Let $\A$ be a forward (\resp{} backward) $\DCMT$. Given a data word $w=(a_1,d_1)\ldots(a_n,d_n)$, 
a successful run $\rho$ of $\A$ on $w$ (a unique one if it
exists) is a sequence of states
$q_0q_1\ldots q_n$  (\resp{} $q_n \ldots q_1q_0$) and the output of the run is a word $a_1'\ldots a_n'$ such that,

\begin{itemize}[-]
\item $q_0$ is the initial state,
\item $q_n$ is a global final state,
\item for any position $i$ which does not have a 
class successor (\resp{} class predecessor), the state $q_i$ (\resp{} $q_{n-i+1}$) is a class final state.
\item Let $i$ be a  position with the types 
$(p,s)\in \M$. Then,
\begin{itemize}
 \item if $i$ has no class predecessor (\resp{} no class successor) 
then the tuple $(q_{i-1}, \bot, a_i, p, s, q_i, a_i')$ (\resp{} $(q_{n-i}, \top, a_{i}, p, s,
q_{n-i+1}, a_{i}')$) is in $\Delta$, and,
\item if has a class predecessor (\resp{} class successor) (say $j$), then the tuple 
$(q_{i-1}, q_j, a_i, p, s, q_i, a_i')$ (\resp{} $(q_{n-i}, q_{n-j+1}, a_{i}, p, s, q_{n-i+1}, a_{i}')$) is in $\Delta$.

\end{itemize}
\end{itemize}

Note that if there is a successful run it is unique and it defines a unique output data word 
$w'$ which is obtained by applying the labelling supplied by the run to the data word $w$ (that is
$\gdp{w'}=\gdp{w}$ and $\gsp{w'}=\A(\msp{w})$).

\begin{definition}
A {\em cascade} of $\DCMT$ $\C$ is a sequence, $\langle \Sigma=\Sigma_0, \A_1, \Sigma_1, \ldots, \Sigma_{n-1}, \A_n, \Sigma_{n}\rangle$ 
such that $\A_1,\ldots, \A_n$ is a sequence of forward and backward $\DCMT$s and for each $i$, $\A_i$ is a $\DCMT$ with input alphabet
$\Sigma_{i-1}$ and output alphabet $\Sigma_{i}$. We denote by $\mathsf{D}$ the set of all cascades of $\DCMT$. 
\end{definition}
The run of $\C$ is defined as before.

\begin{remark}
\label{DCMTunion}
 Using standard product construction it follows that forward (\resp{} backward) $\DCMT$ are closed under product.
This can be extended to cascades.
 \end{remark}

\begin{proposition} 
\mlabel{proposition:br-to-cascade}
For every BR-formula $\varphi$ of $\comp$-height $k$ there is an equivalet cascade in $\mathsf{D}$ of height $k$. 
\end{proposition}
\begin{proof} 
Let us observe that it is sufficient to prove the following claim;
($\star$) {\em for every formula $\varphi$ in $\mathsf{Formulas}(M_{\mathtt{Y}})$ there is a forward $\DCMT$ which outputs
$\varphi$ at every position where it holds in the input.} By symmetry we will obtain that
for every formula $\varphi$ in $\mathsf{Formulas}(M_{\mathtt{X}})$ there is a backward $\DCMT$ which outputs
$\varphi$ at every position where it holds in the input. Since by Remark \ref{DCMTunion} given a finite set of formulas
$\{\varphi_1, \ldots, \varphi_k\}$ we can find a 
forward $\DCMT$ which will label every position of the input with the precise subset of formulas which are true there.
Finally since BR and $\mathsf{D}$ are closed under composition (by induction on height $k$) the proposition follows.

Next we show ($\star$). Without loss of generality assume $\varphi$ is guarded and uses only $\nu$-fixpoints.
Recall the definition of closure and atom (Definitions
\ref{definition:closure} and \ref{definition:atom}).
We define a forward-$\DCMT$ $\A_\varphi$ whose states are precisely the 
atoms in $\fl(\varphi)$. Let us observe that using Corollary \ref{corollary:BR-unique} every formula in every atom in $\fl(\varphi)$ can also be
transformed to
use only $\nu$-fixpoints. 
Next we discuss the transitions of $\A_\varphi$; this machine verifies
that the sequence of atoms defined by the run of the automaton indeed satisfies all consistency conditions defined below.
We let $\(A_{-1}, A_{\minusc 1}, a, p, s, A, a'\)$ to be a transition of $\A_\varphi$ such that
\begin{enumerate}[(i)]
\item if $A_{\minusc 1} = \bot$ then $\firstc \in A$,
\item $a \in A$,
\item $\S \in A$ iff $s=\S$, similarly, $\P \in A$ iff $p=\P$,
\item if  $\varphi \in A_{-1}$ then  $\prevg \varphi\in A$,
\item  if $\varphi\in A_{\minusc 1}$ then $\prevc \varphi  \in A$,
\item $a'=\varphi$ iff $\varphi \in A$.
\end{enumerate}

All states of $\A_\varphi$ are final and class final. The initial state is the unique atom which contains $\firstg$ and $\neg p$ for every
propositional variable $p$. Let us verify that the automaton is deterministic. Assume $\(A_{-1}, A_{\minusc 1}, a, p, s, a\)$ and
$\(A_{-1}, A_{\minusc 1}, a, p, s, A',a'\)$ are two transitions of $\A_\varphi$. We want to show that $A=A'$, that is to say
for every $\psi$ if $\psi\in A \Rightarrow \psi \in A'$ (This is sufficient since atoms are maximal).
We proceed by induction on the structure of $\psi$.
Observe that because of conditions (i--iii) 
whenever $\psi$ is a propositional variable, a zeroary modality or their negation
the claim holds.
When $\psi$ is of the form $\prevg \chi$ (the case of $\prevc \chi$ being similar) then $\psi \in A \Rightarrow
\chi \in A_{-1} \Rightarrow \psi \in A'$ (by conditions (iv--v)). Assume $\psi = \chi \vee \delta \in A 
\Rightarrow \chi \in A \mbox{ or } \delta \in A \Rightarrow \chi \in A' \mbox{ or } \delta \in A' \mbox{ (by IH) } 
\Rightarrow \chi \vee \delta \in A'$. The case of $\wedge$ and $\neg$ is similar.
Finally assume that $\psi = \nu x. \chi(x) \in A$ where $\chi(x)$ is guarded. 
Hence $\chi(\nu x. \chi(x)) \in A$. Let us safely assume (using the unfolding of the fixpoints in the atom) that
$\chi(x)$ is not of the form $\nu y,\chi'(x,y)$. In which case $\chi$ is a boolean combination
of formulas of the form $\mathtt{M}\phi(x)$ or $\phi'$ where $\phi'$ does not contain $x$. We apply induction hypothesis
to $\chi(\nu x. \chi(x))$. For every subformula $\phi'$ of $\chi(x)$, $\phi' \in A \Leftrightarrow  \phi' \in A'$.
For every formula of the form $\mathtt{M}\phi(\nu x. \chi(x))$, it is the case that (by conditions (iv--v)) 
$\mathtt{M}\phi(\nu x. \chi(x)) \in A \Leftrightarrow \mathtt{M}\phi(\nu x. \chi(x)) \in A'$. Hence we conclude that
$\psi \in A'$.

Next we show the correctness of the construction. For a given data word $w$, we observe that the sequence of atoms
$A_0,A_1,\ldots,A_n$ where $A_i$ is the set of all formulas in $\fl(\varphi)$ is an accepting run of $\A_\varphi$.
For the other direction we need to show that if $A_0,A_1,\ldots, A_n$ is the unique accepting run of $\A_\varphi$ on $w$, then
for every formula $\psi \in \fl(\varphi)$, $A_i \ni \psi \Leftrightarrow w,i \models \psi$.
We prove the stronger claim; {\em for every formula $\varphi$ and for every data words $w$ and every sequence $A_0,A_1,\ldots, A_n$
satisfying conditions (i--vi) it is the case that for every formula $\psi \in \fl(\varphi)$, if $A_i \ni \psi \Rightarrow w,i \models \psi$.}
Note that if $A_i \not \ni \psi \Rightarrow A_i \ni \neg \psi \Rightarrow w,i \models \neg \psi \Rightarrow w,i \not \models \psi$.
Proof of the claim is a repetition of the similar claim in the proof of Theorem \ref{theorem:nu-to-data} using Corollary \ref{corollary:BR-unique}.
\end{proof}

\begin{proposition}
\mlabel{proposition:cascade-to-br}
 For every cascade of height $k$ there is an equivalet BR-formula of $\comp$-height $k+1$.
\end{proposition}
\begin{proof}
We prove the following claim; {\em Given a forward $\DCMT$ $\A$ with output alphabet $\Sigma'$ and a letter $a' \in \Sigma'$ 
there is a formula $\varphi_{a'}$ in the composition (of height $2$) of $\mathsf{Formulas}(M_{\mathtt{Y}}) \cup \mathsf{Formulas}(M_{\mathtt{X}})$
such
that $\A$ on input $w$ outputs $a'$ on position $i$ iff
$w,i \models \varphi_{a'}$}. By symmetry we obtain the
analogous claim for backward $\DCMT$. Furthermore since
BR is closed under composition we obtain the claim for cascades of arbitrary height (by induction on the height). 
Finally to check that the cascade accepts the input, all we need to check is that the some output is produced at the first position .

Next we prove the claim. Let $\A$ be a forward $\DCMT$ with set of states $\{q_1,\ldots, q_n\}$ and transitions $\Delta$ and 
initial state $q_1$ and class and global final states $F_c$ and $F_g$ respectively. Let us assume without loss of generality that there are no
incoming transitions to $q_1$. Denote by $\bar{x_q}$ the tuple of variables $x_{q_1},\ldots,x_{q_n}$.
Let $\psi_{q_i}(\bar{x_q})$ be the formula

$$ \psi_{q_i}(\bar{x_{q}}):= \bigvee_{\delta \in
\Delta} \left\{\begin{array}{ll} \(\prevg
x_q \wedge \firstc \wedge a \wedge p \wedge s\) & \mbox{ if  $q' = \bot$}\\
\(\firstg \wedge a \wedge p \wedge s\) & \mbox{ if $q= q_1$}\\
\(\prevg x_q \wedge \prevc x_{q'} \wedge a \wedge p \wedge s\) & \mbox{ else }  \\
\end{array}\right.$$
where $\delta = (q, q', a, p, s, q_i, a')$.

We write a formula in vectorial form (see \cite{ArnoldNiwinski} for related definitions and results)
of the following form, 
$$\varphi=\nu \left(\begin{array}{c}x_{q_1}\\ \vdots\\ x_{q_n} \end{array}\right). \left ( \begin{array}{c} \psi_{q_1}(x_{q_1},\ldots, x_{q_n})\\
\vdots\\ 
\psi_{q_n}(x_{q_1},\ldots, x_{q_n})\\ \end{array}\right )$$
which computes the unique run of the $\DCMT$ (if it exists) as a vector of subsets of positions. 
Now, using Bekic's principle one can linearize this vectorial $\mu$-calculus 
formula to yield a $\mu$-calculus formula $\varphi_{q_i}$ which computes the set of positions $x_{q_i}$ at the fixpoint of $\varphi$. 
Now $\varphi_{a'}$ is defined as
\begin{align*}
\varphi_{a'}:= \( \firstc \rightarrow \vee_{q_i\in F_c} \varphi_{q_i}\) &\wedge \( \firstg \rightarrow \vee_{q_i\in F_g} \varphi_{q_i}\)\\
&\wedge \bigvee_{\delta\in
\Delta} \psi_{q_i}\(\varphi_{q_1},\ldots,\varphi_{q_n}\),
\end{align*}

where $\delta = (q, q', a, p, s, q_i, a')$.

Note that so far the formulas $\varphi_{a'}$ is true at a position $i$ iff the unique partial run outputs $a'$ on it.
For the inductive case this is enough. To assert that there is a successful run we write the formula $\Gg \vee_{a'\in \Sigma'} \varphi_{a'}$ which is
of $\comp$-height $2$.
\end{proof}

% !TEX root =  main.tex

Hence we obtain,
\begin{theorem}
 BR and $\mathsf{D}$ are equivalent.
\end{theorem}

\newcommand{\ign}[1]{}
\ign{
\subsection{Characterizing BR and BMA using cascades} 
The analogue of composition operator $\comp$ in logic is played by cascade in automata. We conclude this section by introducing the cascades of automata
which captures BMA and BR. We keep the discussions below to an informal level, further details can be found in Appendix \ref{rue-de-cascade}.

A global transducer is a functional finite state transducer which sees only the string projection of the data word and outputs another word. 
A class transducer is similarly a functional finite state transducer though it differs in the way it functions. On a given input data word
several copies of the class transducer runs over each class and each copy outputs a word and these words are put together 
(in the original order of positions) to create the new data word. By a cascade of such transducers we mean a sequence of automata $\A_1\ldots, \A_k$
such that the output of $\A_i$ is fed into the input of $\A_{i+1}$.

Now, consider a formula $\varphi$ in BMA. Since it is in BMA it makes boundedly many switches (say k) between the class mode and global mode. 
A cascade $\A_1,\ldots, \A_k$ verifies this formula as follows. The automaton $\A_1$ (assume it is a global transducer) marks each position with the 
subformulas $\varphi$ which uses only global modalities, while the automaton $\A_2$ (assume it is a class transducer) marks each position  
with subformulas which consists of class modalities and formulas verified by $\A_1$. Each successive automata $\A_i$, thus, evaluates subformulas 
whose $\comp$-height (or the number of switches) is at most $i$. This way we prove,

\begin{proposition}\mlabel{bma-cascade}
 BMA and cascades of finite state transducers are equivalent.
\end{proposition}

For characterizing BR the reasoning is similar. Instead of finite state transducers we use class memory transducers which are the tranducer version
of class memory automata from \cite{BjorklundS10}. Here a level of the transduction corresponds to formulas which uses modalities which are of the 
same direction. We prove that,

\begin{proposition}\mlabel{br-cascade}
 BR and cascades of class memory transducers are equivalent.
\end{proposition}
}

%% file: dltl.tex
\section{Data-LTL and $\mathbf{\mbox{FO}^2}$}
\mlabel{section:data-ltl}
Here we make a remark about two logics already discussed in the literature namely $\fotwo$ \cite{BojanczykDMSS11} and Data-LTL \cite{KaraSZ10}. 
Data-LTL (DLTL for short) was introduced in \cite{KaraSZ10} in the setting of data words with multiple data values. 
We restrict it to the case of data words. 
The fragment described below is called Basic DLTL there. It has
the following syntax, let $M_1=\{\nextg, \prevg, \nextc, \prevc\}$ and $M_2=\{ \Ug, \Sg, \Uc, \Sc\}$,
\begin{align*}
\varphi := p \in \mathrm{Prop}\mid \S \mid \P &\mid \mathtt{M_1} \varphi, \mathtt{M_1}\in M_1 \\
                                              &\mid \varphi \wedge \varphi \mid \neg \varphi \mid \varphi \mathtt{M_2} \varphi,
\mathtt{M_2} \in M_2\ .
\end{align*}
From the Example \ref{modalityexample} it is clear that DLTL is a subclass of BMA.
The fragment of DLTL containing the set of modalities $\{\nextg, \nextc, \prevg, \prevc, \Fc, \Fg, \Pg, \Pc\}$ is called unary-Data-LTL. 

Define the modalities $\mathtt{fF}^{\not \sim}$  (\textit{far-future not in class}) and $\mathtt{dP}^{\not \sim}$
(\textit{deep-past not in class}) as,

$$\begin{array}{lll}
 w,i \models \mathtt{fF}^{\not \sim} \varphi &\Leftrightarrow& \exists j>i+1 \mbox{ such that } i \not \sim j \mbox{ and } w,j \models \varphi\\
 w,i \models \mathtt{dP}^{\not \sim} \varphi &\Leftrightarrow& \exists j<i-1 \mbox{ such that } i \not \sim j  \mbox{ and } w,j \models \varphi\\
\end{array}
$$
\begin{lemma}
The modalities $\mathtt{fF}^{\not \sim}$ and $\mathtt{dP}^{\not \sim}$ are  expressible using the modalities 
$\{\nextg, \nextc, \prevg, \prevc, \Fc, \Fg, \Pg, \Pc\}$ over data words and data $\omega$-words.
\mlabel{lemma:modality-to-ltl}
\end{lemma}
\begin{proof}

{\bf Finite data word case:}
We only do the case of $\mathtt{fF}^{\not \sim}$. The case of $\mathtt{dP}^{\not \sim}$ is symmetric. Assume we are given a formula 
$\mathtt{fF}^{\not \sim} \varphi$. Let $k$ be the 
last position where $\varphi$ is true. Obviously it is the unique position where $ \varphi_{\mathit{last}} = \varphi \wedge \neg \Fg \varphi$ is
true. 
A position $i$ satisfies $\mathtt{fF}^{\not \sim} \varphi$ if and only if one of the following scenarios hold;
\begin{enumerate}
 \item $k>i+1$ and $k \not \sim i$,
 \item $k \sim i$ and there is a $j>i+1$ such that $j$ satisfies $\varphi$ and $j \not \sim k$.
\end{enumerate}
The first scenario holds if the formula $\nextg\nextg\Fg \varphi_{\mathit{last}} \wedge \neg \Fc \varphi_{\mathit{last}}$ is true at position $i$. 
(Note that $\Fg$ evaluates a formula on all positions in the future including the current position, hence $\nextg\nextg\Fg \varphi_{\mathit{last}}$).
The second scenario holds if the formula $\Fc \varphi_{\mathit{last}} \wedge \nextg\nextg\Fg(\varphi \wedge \neg \Fc \varphi_{\mathit{last}})$ holds
at 
position $i$. Hence $\mathtt{fF}^{\not \sim} \varphi$ is equivalent to the formula
$$\Psi \equiv (\nextg\nextg\Fg \varphi_{\mathit{last}} \wedge \neg \Fc \varphi_{\mathit{last}}) \vee (\Fc \varphi_{\mathit{last}} \wedge
\nextg\nextg\Fg(\varphi
\wedge \neg \Fc \varphi_{\mathit{last}})).$$
\newline
{\bf Data $\omega$-word case:}
Let $\alpha$ be a data $\omega$-word and $i$ be a position of $\alpha$.
Below we characterize the scenarios when $i$ satisfies the formula $\varphi$. 
We do a case analysis based on the number of classes in $\alpha$ which has infinitely many 
positions satisfying $\varphi$. 

{\bf case 1: when all classes of $\alpha$ has only finitely many positions satisfying $\varphi$} :
Let us observe that this is the case if and only if all class minimum positions in $\alpha$ satisfy the formula
$\Fc\Gc \neg \varphi$ . Hence $\alpha$ belongs to this case if and only if $\alpha$ satisfies the formula
$$C_1 \equiv \firstg \rightarrow \Gg\( \firstc \rightarrow \Fc\Gc \neg \varphi \) .$$
In this scenario we have two subcases;

{\it subcase 1: When there are only finitely many $\varphi$ in $\alpha$ :} This is the case if and 
only if $\alpha$ satisfy the formula $$S_1 \equiv \firstg \rightarrow \Fg \Gg \neg \varphi \,.$$
Note that in thie case our reasoning essentially is the same as that of the finite data word case.
Hence in this subcase a position $i$ satisfies $\mathtt{fF}^{\not \sim} \varphi$ if and only if 
it satisfies the formula 
$$ \Phi_1 \equiv \Hg \( C_1 \wedge S_1 \) \rightarrow \Psi \,.$$
{\it subcase 2: When there are infinitely many $\varphi$ in $\alpha$ :} This is the case if and 
only if $\alpha$ satisfies the formula 
$$S_2 \equiv \firstg \rightarrow \Gg \Fg \varphi \,.$$
Also observe that since all classes in $\alpha$ contain only finitely many $\varphi$ and $\alpha$ contain
infinitely positions with $\varphi$ it is the case that there are infinitely many classes in 
$\alpha$ containing a $\varphi$. Therefore it is guaranteed that all positions $i$ have a 
position to the right which is not in its class and which satisfies $\varphi$. We can characterize
this subcase by the formula 
$$\Phi_2 \equiv \Hg\(C_1 \wedge S_2\) \rightarrow \true \,.$$

{\bf case 2: when there is exactly one class in $\alpha$ which has  infinitely many positions satisfying $\varphi$} : 
First we observe that we can characterize this case using a formula. This scenario holds
if in $\alpha$ there is exactly one class minimum posiiton satisfying the formula $\Gc\Fc \varphi$ and 
all other class minimum position satisfies the formula $\Fc\Gc \neg \varphi$. 
Therefore the positions in the unique class (call it $I$) containing infinitely many $\varphi$ are characterized
by the formula 

\begin{align*} U \equiv \Pc \( \firstc \wedge \Gc\Fc \varphi \wedge \nextg\Fg \(\firstc \rightarrow \Fc \Gc \neg \varphi\)\right.  \\
\quad\left. \wedge \prevg\Pg \(\firstc \rightarrow \Fc \Gc \neg \varphi\) \)\, .
\end{align*}

Using the formula $U$ we can assert that $\alpha$ belongs this class by stating that 
$\Fg U$. Now observe that in this scenario a position $i$ satisfies the formula $\mathtt{fF}^{\not \sim} \varphi$ if and only if one of the following
two conditions hold;
\begin{enumerate}
 \item $i$ is not in the class $I$, which is encoded by the formula $\neg U$,
 \item $i$ is in the class $I$ and there is a $j > i+1$ such that $j$ satisfies $\varphi$ and $j$ is not in $I$. This is encoded by the formula 
 $$ U \wedge \nextg \nextg \Fg \(\varphi \wedge \neg U\) \,.$$  
\end{enumerate}
Hence in this case we can say that $\mathtt{fF}^{\not \sim} \varphi$ is equivalent to the formula
$$\Phi_3 \equiv \Fg U \vee \Pg U \rightarrow \(\neg U \vee \(U \wedge \nextg \nextg \Fg \(\varphi \wedge \neg U\)\)\)\, .$$

{\bf case 3: when there are atleast two classes in $\alpha$ containing infinitely many positions satisfying $\varphi$} : If this is the case then
every position in $\alpha$ satisfies the formula $\mathtt{fF}^{\not \sim} \varphi$. We can 
check this case by stating that there exist two class minimum positons where the formula
$\Gc\Fc \varphi$ holds. Hence in this case $\mathtt{fF}^{\not \sim} \varphi$ is equivalent to 
the formula

\begin{align*}
 \Phi_4 \equiv \Pg \( \firstg \wedge  \Fg \( \firstc \right. \right.& \\
  &\hspace{-1cm}\left. \left. \wedge \(\Gc\Fc \varphi \wedge \nextg\Fg \( \firstc \wedge \(\Gc\Fc \varphi \)\)\)\)\)\, .\\
\end{align*}

Finally to conclude the proof we observe that the three cases described above are exhaustive and hence the formula $\mathtt{fF}^{\not \sim} \varphi$
is equivalent to the disjunction 
$$\Phi_1 \vee \Phi_2 \vee \Phi_3 \vee \Phi_4\, .$$
\end{proof} 

\begin{corollary}
The modalities $\mathtt{F}^{\not \sim}$ (future not in class) and $\mathtt{P}^{\not \sim}$ (past not in
class) defined as 
$$\begin{array}{lll}
 w,i \models \mathtt{F}^{\not \sim} \varphi &\Leftrightarrow& \exists j>i \mbox{ such that } i \not \sim j \mbox{ and } w,j \models \varphi\\
 w,i \models \mathtt{P}^{\not \sim} \varphi &\Leftrightarrow& \exists j<i \mbox{ such that } i \not \sim j  \mbox{ and } w,j \models \varphi\\
\hfill\mathtt{G}^{\not \sim} \varphi & \Leftrightarrow & \neg \mathtt{F}^{\not \sim} \neg \varphi \\
\hfill\mathtt{H}^{\not \sim} \varphi & \Leftrightarrow & \neg \mathtt{P}^{\not \sim} \neg \varphi \\
\end{array}
$$ is definable in DLTL over data words and data $\omega$-words.
\end{corollary}
\begin{proof} 
Define $\mathtt{F}^{\not \sim} \varphi \equiv \( \neg \S \wedge \nextg \varphi\) \vee \mathtt{fF}^{\not \sim} \varphi$ and 
$\mathtt{P}^{\not \sim} \varphi \equiv \( \neg \P \wedge \prevg \varphi\) \vee \mathtt{dP}^{\not \sim} \varphi$.
\end{proof}

\begin{remark} In \cite{DL08} it is shown that $\fotwo(\Sigma, <, +1, \sim)$ and \intro{simple freeze-LTL} (LTL with the operators $\downarrow,
\uparrow$ for the
registers 
and modalities $\nextg, \Fg,\prevg, \Pg$ and their duals such that the each modality is immediately preceded by a freeze $\downarrow$) are equivalent.
Applying the above idea 
it follows that formulas in simple freeze-LTL can be equivalently written such that the negation appears only at the propositional 
variables (i.e. no need to have negation at the de-freeze operator, i.e. no need to have $\uparrow \not \sim$).
\end{remark}

Next using the above lemma, we prove the equivalence between $\fotwo\(\Sigma, <, +1, \sim, \succclass\)$ and DLTL. The \intro{modal-depth} of a
DLTL formula and the \intro{quantifier-depth} of an $\fotwo$ formula are defined as the maximum number of nested modalities and the 
maximum number of nested quantifiers in the formula.

\begin{theorem}
\label{theorem:udltl-fo2}
$\fotwo\(\Sigma, <, +1, \sim, \succclass\)$ and unary-DLTL are equivalent over data words and 
data $\omega$-words\footnote{It is known from
\cite{KaraSZ10}(Proposition 2) that unary-Data-LTL extended with the additional modalities 
$P^{\not\sim}$ and $F^{\not\sim}$ is equivalent to $\fotwo\(\Sigma, <, +1, \sim, \succclass\)$. However this result uses fewer modalities and is not
known before.}\!. More precisely, 
\begin{enumerate}
 \item  for every unary-DLTL formula $\varphi$ there is a
$\fotwo\(\Sigma, <, +1, \sim,\succclass\)$ formula $\varphi'(x)$ such that $w,i \models \varphi$ if and only if $w,i \models \varphi'(x)$.
Moreover the size of $\varphi'(x)$ is linear in the size of the formula. Similarly the quantifier-depth of $\varphi'(x)$ is the same as the 
modal-depth of of $\varphi'$. 
\item Similarly,
for every $\fotwo\(\Sigma, <, +1, \sim, \succclass\)$ formula $\varphi(x)$ there is a unary-DLTL formula $\varphi'$ such that 
$w,i \models \varphi'$ if and only if $w,i \models \varphi(x)$. The size of $\varphi'$ is exponential in the size of $\varphi(x)$. 
The modal-depth of $\varphi'$ is linear in the quantifier-depth of $\varphi(x)$. 
\end{enumerate}
\end{theorem}
\begin{proof} 

($\Leftarrow$)
 Follows simply from the fact that the modalities used in unary-DLTL are expressible in $\fotwo\(\Sigma, <, +1, \sim, \succclass\)$ and we use the 
 obvious analogue of the standard translation from modal logic to two-variable first order logic. The translation is linear and preserves the depth 
 as claimed. 

($\Leftarrow$)

For convenience we define the  abbreviations $x \ll y$ and $x \ll^c y$ for $x < y \wedge x+1 \neq y$ and  
$x \sim y \wedge x < y \wedge x \succclass \neq y$.

We intend to prove that for every $\fotwo\(\Sigma, <, +1, \sim, \succclass\)$ formula $\varphi(x)$ there is a unary-DLTL formula $\varphi'$ such that 
$w,i \models \varphi'$ if and only if $w,i \models \varphi(x)$. The proof idea is quite standard (see \cite{EtessamiVW02}).
Let $\varphi(x)$ be a formula in $\fotwo$, the {\em quantifier depth} of $\varphi(x)$ 
is defined as usual as the maximum number of nested quantifiers in $\varphi(x)$. The proof is by induction on the structure of the formula. 
When $\varphi(x)$ is $a(x)$ then $\varphi'$ is simply $a$. When $\varphi(x)$ is of the form $\varphi_1(x) \vee \varphi_2(x)$ (or $\neg \varphi_1(x)$),
using inductive hypothesis, we define $\varphi'$ as $\varphi_1' \vee \varphi_2'$ (or $\neg \varphi_1'$). The remaining cases are that when 
$\varphi(x)$ is of the form $\exists x. \varphi_1(x)$ or $\exists y. \varphi_1(x,y)$. Both cases are identical upto a renaming of variables. So it is
enough to consider only $\exists y. \varphi(x,y)$. We write $\varphi(x,y)$ in disjunctive normal form and distribute the existential quantifier over
the 
disjunctions to obtain a formula of the form $\bigvee_i \exists y. \varphi_i(x,y)$ where each $\varphi_i(x,y)$ is of the form 
$\alpha_i(x) \wedge \beta_i(y) \wedge \delta_i(x,y) \wedge \gamma_i(x,y)$ in which $\alpha_i(x), \beta_i(y)$ are formulas with only one free variable
$\delta_i(x,y) \in \Delta(x,y)$ and $\gamma_i(x,y) \in \Gamma(x,y)$ 
where the sets $\Delta(x,y)$ and $\Gamma(x,y)$ are, 

$$\Delta(x,y)= \{ y \ll x,\, y+1 =x,\, x =y ,\, x + 1 = y,\, x \ll y \},$$
$$\Gamma(x,y) = \{ y \ll^c x,\, y\succclass = x,\, x \not \sim y,\, x \succclass =y,\,  x \ll^c y \}.$$

We also note that writing each conjuct $\varphi_i$ in this form might require replacing subformulas in $\varphi_i$ which are 
negations of formulas in $\Delta(x,y)$ by an equivalent formula consisting of disjunctions of formulas from $\Delta(x,y)$ 
(and further distributing these disjunctions in the conjunct). Let us observe that it is enough to define a translation for each of the disjunct 
of the form $\varphi(x,y) \equiv \exists y.\, \alpha(x) \wedge \beta(y) \wedge \delta(x,y) \wedge \gamma(x,y)$. 
Inductively we assume that we have the DLTL formulas $\alpha'$ and $\beta'$ which are equivalent to $\alpha(x)$ and $\beta(y)$.
We define the translation below.

Consider the case when $\gamma(x,y)$ is $x \not \sim y$. Then the translations are listed below.

$$
\arraycolsep=1.5pt
\begin{array}{l@{\hspace{.5cm}}|@{\hspace{.5cm}}l}
 \delta(x,y) & \varphi'\vspace{1pt}\\
 \hline
 x = y & \mathsf{false}\\
x\ll y &  \alpha' \wedge \mathtt{fF}^{\not \sim} \beta' \\
x+1=y & \alpha' \wedge \neg \S \wedge \nextg \beta' \\
y+1=x &  \alpha' \wedge \neg \P \wedge \prevg \beta' \\
y \ll x & \alpha' \wedge \mathtt{dP}^{\not \sim} \beta'\\
\end{array}$$

The rest of the cases are symmetric and hence we treat only the cases when $x \leq y$. 

Assume $\delta(x,y)=x \ll y$. Then $\varphi(x,y)$ is satisfiable only when $\gamma(x,y)$ is $x \succclass = y$, $x \ll^c y$ and we define 
the respective translations as $\alpha \wedge \nextc \beta' \wedge \neg \S$ and $\alpha \wedge \nextc \nextc \Fc \beta'$.

When $\delta(x,y)$ is $x +1 = y$, $\varphi(x,y)$ is satisfiable only when $\gamma(x,y)$ is $x \succclass= y$ and we define the translation
as $\alpha \wedge \nextc \beta' \wedge \S$.

For estimating the size and modal depth one proceeds by induction. We omit the analysis as it is straightforward.
\end{proof}

Finally, remark that the separation of LTL and unary-LTL over words implies that (consider data words in which all data values are identical)
unary-DLTL is a strictly less expressive than DLTL. Similarly the separation of $\mu$-calculus and LTL over words implies that
DLTL is strictly less expressive than BMA.

%% file: conclusion.tex
% !TEX root =  main.tex

\section{Discussion}
Over data words we have the following inclusions.
\begin{align*}
 \fotwo=\mbox{uDLTL} \stackrel{1}{\subsetneq} \mbox{DLTL} \stackrel{2}{\subsetneq} \mbox{BMA} \stackrel{3}{\subsetneq} \mbox{BR}
\stackrel{4}{\subsetneq} \nu\mbox{-Fragment} \stackrel{5}{\subseteq} \mbox{DA}
\end{align*}
Over data $\omega$-words we have the following inclusions.
\begin{align*}
 \fotwo=\mbox{uDLTL} \stackrel{1'}{\subsetneq} \mbox{DLTL} \stackrel{2'}{\subsetneq} \mbox{BMA} \stackrel{3'}{\subsetneq}
\mbox{DA}\stackrel{4'}{\supseteq} \nu\mbox{-Fragment}
\end{align*}

Inclusions $1$,$1'$,$2$ and $2'$ follow from Example \ref{modalityexample} and Theorem \ref{theorem:udltl-fo2} while the strictness of the
inclusions follow from the respective strictness on words and $\omega$-words (which are data words and data $\omega$-words when $\D$ is
singleton). Inclusion $3$ and $3'$ follows from Theorem \ref{theorem:bma-to-br} and Remark \ref{remark:bma-to-da} while the strictness of the
inclusion depends on deep results
from additive combinatorics which will appear in a later publication. Inclusion $4$ follows from Theorem \ref{theorem:br-to-nu} while
strictness follows from the fact that BR is closed under complementation while $\nu$-fragment is not (Theorem \ref{theorem:undecidability}).
Inclusions $4'$ and $5$ follow from Theorem \ref{theorem:nu-to-data}. The strictness is open.
Also note that over data $\omega$-words $\nu$-fragment has non-empty intersection with uDLTL but do not contain it. The non-containment
follows from the non-containment of unary-LTL in the $\nu$-fragment of $\mu$-calculus on $\omega$-words.

\section{Conclusions}

In this paper we have studied the expressive power of $\mu$-calculus over data words. Though the general logic is undecidable, we disclose
several fragments that are: the $\nu$-fragment, the Bounded Reversal fragment (BR) and the Bounded Mode Alternation fragment (BMA).
BR and BMA happen to form Boolean algebras making them very natural, and relatively expressive logics over data words. We also establish the
relationship with earlier logics like $\fotwo$ or Data-LTL. We end with the following question.

\begin{question}
Cascades of finite state automata can be characterized as wreath product of semigroups (Krohn-Rhodes theorem), a result which
has an 
analogue on trees \cite{bojanczyk2009wreath}. Is there a generalization to BMA? 
\end{question}